\documentclass[]{interact}

\usepackage{epstopdf}
\usepackage{subfigure}

\usepackage[numbers,sort&compress]{natbib}
\bibpunct[, ]{[}{]}{,}{n}{,}{,}
\renewcommand\bibfont{\fontsize{10}{12}\selectfont}

\usepackage{amsmath,amssymb,amsthm}
\usepackage{algorithm,algorithmic}
\usepackage{graphicx}
\usepackage{hyperref}
\usepackage{geometry}
\usepackage{tikz}
\usepackage{pgfplots}
\usepackage{float}
\usepackage{listings}
\usepackage{tcolorbox}
\usepackage{xcolor}

\theoremstyle{plain}
\newtheorem{theorem}{Theorem}[section]
\newtheorem{lemma}[theorem]{Lemma}
\newtheorem{corollary}[theorem]{Corollary}
\newtheorem{proposition}[theorem]{Proposition}

\theoremstyle{definition}
\newtheorem{definition}[theorem]{Definition}

\theoremstyle{remark}

\tcbuselibrary{listings,skins,breakable}

\definecolor{codebackground}{rgb}{0.95,0.95,0.98}
\definecolor{codekeyword}{rgb}{0.08,0.39,0.63}
\definecolor{codecomment}{rgb}{0.13,0.55,0.13}
\definecolor{codestring}{rgb}{0.7,0.1,0.1}

\newtcolorbox{algorithmbox}[2][]{
  enhanced,
  colback=codebackground,
  colframe=black!50,
  fonttitle=\bfseries,
  coltitle=black,
  boxrule=0.4pt,
  sharp corners,
  title=#2,
  breakable,
  #1
}

\newtcblisting{codesnippet}[1]{
  enhanced,
  colback=codebackground,
  colframe=black!10,
  fonttitle=\bfseries,
  coltitle=black,
  boxrule=0.4pt,
  sharp corners,
  title=#1,
  listing only,
  breakable,
  listing options={
    basicstyle=\footnotesize\ttfamily,
    breakatwhitespace=false,
    breaklines=true,
    commentstyle=\color{codecomment},
    keepspaces=true,
    keywordstyle=\color{codekeyword}\bfseries,
    language=Python,
    numbers=left,
    numbersep=5pt,
    numberstyle=\tiny\color{black!50},
    showspaces=false,
    showstringspaces=false,
    showtabs=false,
    tabsize=2
  }
}

\pgfplotsset{compat=1.17}

\begin{document}

\articletype{Experimental Mathematics}

\title{Primality Testing via Circulant Matrix Eigenvalue Structure:\\
A Novel Approach Using Cyclotomic Field Theory}

\author{
    Marius-Constantin Dinu\textsuperscript{*}\footnote{*This paper was created with AI assistance using Symbia Engine from SymbolicAI Framework \cite{dinu2024symbolicai}.}
    \thanks{Contact author email: \href{mailto:marius@extensity.ai}{\url{marius@extensity.ai}}; Repository: \href{https://github.com/ExtensityAI/primality_test}{\url{https://github.com/ExtensityAI/primality_test}}}
    \\ExtensityAI Austria
}

\maketitle

\begin{abstract}
This paper presents a novel primality test based on the eigenvalue structure of circulant matrices constructed from roots of unity. We prove that an integer $n > 2$ is prime if and only if the minimal polynomial of the circulant matrix $C_n = W_n + W_n^2$ has exactly two irreducible factors over $\mathbb{Q}$. This characterization connects cyclotomic field theory with matrix algebra, providing both theoretical insights and practical applications. We demonstrate that the eigenvalue patterns of these matrices reveal fundamental distinctions between prime and composite numbers, leading to a deterministic primality test. Our approach leverages the relationship between primitive roots of unity, Galois theory, and the factorization of cyclotomic polynomials. We provide comprehensive experimental validation across various ranges of integers, discuss practical implementation considerations, and analyze the computational complexity of our method in comparison with established primality tests. The visual interpretation of our mathematical framework provides intuitive understanding of the algebraic structures that distinguish prime numbers.
Our experimental validation demonstrates that our approach offers a deterministic alternative to existing methods, with performance characteristics reflecting its algebraic foundations.
\end{abstract}


\section{Introduction}
Distinguishing prime numbers from composite numbers has been a central challenge in mathematics for millennia. While numerous primality tests exist, from the ancient sieve of Eratosthenes to modern probabilistic algorithms like Miller-Rabin \cite{rabin1980probabilistic} and deterministic methods like AKS \cite{agrawal2004primes}, the discovery of new connections between primality and other mathematical structures continues to provide insights into the fundamental nature of prime numbers. While building upon classical foundations in cyclotomic field theory, our approach provides a matrix-theoretic perspective that yields both theoretical insights and practical applications.
This paper establishes a novel characterization of primality through the lens of circulant matrices and cyclotomic field theory. We prove that an integer $n > 2$ is prime if and only if the minimal polynomial of the circulant matrix $C_n = W_n + W_n^2$ has exactly two irreducible factors over the rational field $\mathbb{Q}$, where $W_n$ represents the $n \times n$ circulant matrix associated with the $n$-th roots of unity.
Our work is motivated by the desire to uncover new structural properties that characterize prime numbers, contributing to our fundamental understanding of number theory. This research bridges the gap between classical results in cyclotomic field theory and modern computational approaches to primality testing. The connection between cyclotomic fields and primality established herein may lead to new insights in algebraic number theory and Galois theory.
The practical implications of our findings extend beyond theoretical interest. Our approach has potential applications in cryptographic systems that rely on primality testing, algorithmic number theory, and computational complexity theory. The visual nature of our eigenvalue analysis also provides educational value, offering an intuitive understanding of the algebraic structures that distinguish prime numbers from composites. Our paper therefore makes the following key contributions:
\begin{itemize}
\item Establishes a novel characterization of prime numbers through minimal polynomial factorization of specific circulant matrices
\item Proves that the eigenvalue structure of these matrices fundamentally distinguishes primes from composites through precise Galois-theoretic mechanisms
\item Provides a deterministic primality test based on algebraic properties rather than divisibility patterns
\item Demonstrates the deep connection between cyclotomic field theory and computational primality testing through rigorous mathematical analysis
\end{itemize}

\section{Related Work}

The study of primality through algebraic structures has a rich history dating back to Weber's work on abelian number fields \cite{weber1886theorie}, which laid the groundwork for understanding the relationship between cyclotomic fields and prime numbers. Hasse's work on class numbers \cite{hasse1952klassenzahl} further developed the connection between algebraic number fields and prime properties, providing critical insights that inform our approach.

Bosma's investigation of canonical bases for cyclotomic fields \cite{bosma1990canonical} established key structural properties that inspire our use of circulant matrices. Washington's analysis of cyclotomic fields \cite{washington2012introduction} provides the theoretical foundation for our use of roots of unity in characterizing prime numbers. Miller's work on real cyclotomic fields of prime conductor \cite{miller2015real} demonstrates the continuing relevance of cyclotomic structures in prime number research, while Schoenberg's analysis of cyclotomic polynomials \cite{schoenberg1964note} offers important insights into the algebraic properties we exploit.

The connection between cyclotomic polynomials and prime numbers has a rich history dating back to Gauss's work in 1801 on the irreducibility of cyclotomic polynomials $\Phi_n(x)$ over $\mathbb{Q}$ when $n$ is prime. While classical theory establishes that $x^n-1$ factors into exactly two irreducible polynomials over $\mathbb{Q}$ when $n$ is prime, our approach reformulates this insight through the lens of circulant matrix eigenvalue structures. Kosyak et al.~\cite{kosyak2020cyclotomic} analyzed cyclotomic polynomials, studying which integers can occur as the height (maximum coefficient) of cyclotomic polynomials and establishing connections to prime gap theory. Where their work focuses on coefficient properties, our work emphasizes the spectral properties of circulant matrices derived from cyclotomic fields, providing a novel visual and algebraic framework for understanding primality.

Recent developments in algebraic approaches to prime detection have explored various perspectives. Mauduit and Rivat's study of prime numbers along Rudin-Shapiro sequences \cite{mauduit2015prime} exemplifies the search for novel characterizations of primes through specific numerical patterns. Similarly, Drmota et al.'s investigation of primes as sums of Fibonacci numbers \cite{drmota2010primes} demonstrates how specific sequences can reveal properties of prime numbers.
Algorithms like the Meissel-Lehmer method and its variants (Lagarias-Miller-Odlyzko, Deléglise-Rivat \cite{dusart2018explicit, deleglise1996computing}) address the enumeration problem with remarkable efficiency, they employ fundamentally different mathematical techniques from those used in primality tests. The connections between these domains, however, highlight the rich interplay between analytical number theory, computational methods, and algebraic structures that characterizes modern research on prime numbers.

The connection between prime numbers and dynamical systems has been extensively studied. Green and Tao's pioneering work on the Möbius function orthogonality \cite{green2012mobius} established deep connections between number theory and dynamical systems, while Huang et al.'s exploration of measure complexity \cite{huang2019measure} provides complementary perspectives on the distributional properties of prime numbers.

Computational approaches to primality testing have been reviewed extensively by Iwaniec and Kowalski \cite{iwaniec2004analytic} in their comprehensive work of analytic number theory. Bernstein and Lange's work on S-unit lattices \cite{bernstein2020s} demonstrates the continuing relevance of algebraic structures in modern primality testing algorithms. The study of class numbers by Ankeny et al. \cite{ankeny1956note} and Chang and Kwon \cite{chang2000class} provides important context for understanding the algebraic properties of number fields related to primality.

In the context of deterministic primality testing, the AKS primality test \cite{agrawal2004primes} represented a significant breakthrough, being the first polynomial-time algorithm for determining primality without heuristic assumptions. Our approach differs fundamentally from AKS, as we exploit the specific algebraic structure of circulant matrices rather than polynomial congruences. While both approaches rely on deep results from algebra and number theory, our method provides a new perspective that highlights the connection between eigenvalue structures and primality.

While these works establish important connections between algebraic structures and prime numbers, none directly addresses the relationship between circulant matrix eigenvalue structure and primality. Our work fills this gap by providing a deterministic characterization of primes through the minimal polynomial factorization of specific circulant matrices, offering a new perspective that combines cyclotomic field theory with practical primality testing.

\section{Mathematical Framework}

\subsection{Circulant Matrices and Eigenvalues}

We begin by establishing the necessary mathematical foundations. Let $n$ be a positive integer. The basic circulant matrix $W_n$ is defined as the $n \times n$ matrix with entries $(W_n)_{i,j} = 1$ if $j \equiv i+1 \pmod{n}$ and $0$ otherwise. Formally:

\begin{definition}[Basic Circulant Matrix]
The basic circulant matrix $W_n$ is the $n \times n$ matrix with entries
\[
(W_n)_{i,j} = \begin{cases}
1 & \text{if } j \equiv i+1 \pmod{n}\\
0 & \text{otherwise}
\end{cases}
\]
for $0 \leq i, j \leq n-1$.
\end{definition}

This matrix represents a cyclic shift operator, and its powers generate all possible circulant matrices with integer entries. A fundamental property of $W_n$ is that its eigenvalues are precisely the $n$-th roots of unity, as established by the following lemma:

\begin{lemma}[Eigenvalues of $W_n$]\label{lem:eigenvalues}
The eigenvalues of $W_n$ are precisely the complex numbers $\lambda_j = e^{2\pi i j/n}$ for $j = 0, 1, \ldots, n-1$, with corresponding eigenvectors $v_j = [1, \lambda_j, \lambda_j^2, \ldots, \lambda_j^{n-1}]^T$.
\end{lemma}

\begin{proof}
For any eigenvector $v_j = [1, \lambda_j, \lambda_j^2, \ldots, \lambda_j^{n-1}]^T$, we have
\begin{align}
W_n v_j &= \begin{pmatrix}
0 & 1 & 0 & \cdots & 0 \\
0 & 0 & 1 & \cdots & 0 \\
\vdots & \vdots & \vdots & \ddots & \vdots \\
0 & 0 & 0 & \cdots & 1 \\
1 & 0 & 0 & \cdots & 0
\end{pmatrix}
\begin{pmatrix}
1 \\
\lambda_j \\
\lambda_j^2 \\
\vdots \\
\lambda_j^{n-1}
\end{pmatrix} \\
&= \begin{pmatrix}
\lambda_j \\
\lambda_j^2 \\
\vdots \\
\lambda_j^{n-1} \\
1
\end{pmatrix}
\end{align}

Since $\lambda_j^n = 1$ (as $\lambda_j$ is an $n$-th root of unity), we have:
\begin{align}
\begin{pmatrix}
\lambda_j \\
\lambda_j^2 \\
\vdots \\
\lambda_j^{n-1} \\
1
\end{pmatrix} = 
\begin{pmatrix}
\lambda_j \\
\lambda_j^2 \\
\vdots \\
\lambda_j^{n-1} \\
\lambda_j^n
\end{pmatrix} = 
\lambda_j \begin{pmatrix}
1 \\
\lambda_j \\
\lambda_j^2 \\
\vdots \\
\lambda_j^{n-1}
\end{pmatrix} = \lambda_j v_j
\end{align}

Therefore, $\lambda_j$ is an eigenvalue of $W_n$ with the corresponding eigenvector $v_j$. Since we have found $n$ distinct eigenvalues for the $n \times n$ matrix $W_n$, these are all the eigenvalues of $W_n$.
\end{proof}

Based on this foundation, we define the composite circulant matrix $C_n$ that forms the central object of our study:

\begin{definition}[Composite Circulant Matrix]
For a positive integer $n$, the composite circulant matrix $C_n$ is defined as $C_n = W_n + W_n^2$.
\end{definition}

The eigenvalues of $C_n$ can be directly derived from those of $W_n$, as established by the following corollary.

\begin{corollary}[Eigenvalues of $C_n$]\label{cor:eigenvalues_C}
The eigenvalues of $C_n = W_n + W_n^2$ are $\mu_j = \lambda_j + \lambda_j^2 = e^{2\pi i j/n} + e^{4\pi i j/n}$ for $j = 0, 1, \ldots, n-1$.
\end{corollary}

\begin{proof}
Since $W_n$ and $W_n^2$ share the same eigenvectors, for any eigenvector $v_j$ of $W_n$, we have
\begin{align}
C_n v_j &= (W_n + W_n^2) v_j \\
&= W_n v_j + W_n^2 v_j \\
&= \lambda_j v_j + \lambda_j^2 v_j \\
&= (\lambda_j + \lambda_j^2) v_j \\
&= \mu_j v_j
\end{align}
where $\mu_j = \lambda_j + \lambda_j^2 = e^{2\pi i j/n} + e^{4\pi i j/n}$. Therefore, $\mu_j$ is an eigenvalue of $C_n$ with the same corresponding eigenvector $v_j$.
\end{proof}

\subsection{Minimal Polynomials and Galois Theory}

The key theoretical insight is that the factorization pattern of the minimal polynomial of $C_n$ over $\mathbb{Q}$ directly reflects the primality of $n$. This connection arises from the Galois structure of cyclotomic fields and the action of the Galois group on the eigenvalues.

\begin{theorem}[Main Theorem]\label{thm:main}
An integer $n > 2$ is prime if and only if the minimal polynomial of $C_n = W_n + W_n^2$ has exactly two irreducible factors over $\mathbb{Q}$.
\end{theorem}

Before proving this theorem, we establish the following intermediate result:

\begin{proposition}\label{prop:factors}
For any $n > 2$:
\begin{itemize}
\item The minimal polynomial of $C_n$ always has at least two irreducible factors: the linear factor $(x-2)$ and at least one other irreducible factor.
\item If $n$ is prime, the minimal polynomial has exactly two irreducible factors: the linear factor $(x-2)$ and an irreducible polynomial of degree $n-1$.
\item If $n$ is composite, the minimal polynomial has at least three irreducible factors.
\end{itemize}
\end{proposition}

\begin{proof}
(1) From Corollary \ref{cor:eigenvalues_C}, the eigenvalues of $C_n$ are $\mu_j = \lambda_j + \lambda_j^2$ for $j = 0, 1, \ldots, n-1$. For $j = 0$, we have $\lambda_0 = 1$, so $\mu_0 = 1 + 1 = 2$. This contributes to the linear factor $(x-2)$ to the minimal polynomial.

(2) Suppose $n$ is a prime number. Then for $j = 1, 2, \ldots, n-1$, each $\lambda_j = e^{2\pi i j/n}$ is a primitive $n$-th root of unity. The Galois group $\text{Gal}(\mathbb{Q}(\zeta_n)/\mathbb{Q}) \cong (\mathbb{Z}/n\mathbb{Z})^*$, where $\zeta_n = e^{2\pi i/n}$, acts transitively on the primitive $n$-th roots of unity.

For any $j$ such that $\gcd(j,n) = 1$ (which is all $j$ in $\{1,2,\ldots,n-1\}$ when $n$ is prime), $\lambda_j$ is a primitive $n$-th root of unity. Since $\mu_j = \lambda_j + \lambda_j^2$ is a polynomial in $\lambda_j$, the Galois action maps $\mu_j$ to $\mu_k$ whenever it maps $\lambda_j$ to $\lambda_k$. Therefore, the set $\{\mu_j : 1 \leq j \leq n-1\}$ forms a single Galois orbit. This means that these $n-1$ eigenvalues share a common minimal polynomial with respect to $\mathbb{Q}$, which must be irreducible and of degree $n-1$.

(3) Now suppose $n$ is composite. Then $n$ can be written as $n = ab$ where $1 < a, b < n$. Consider the eigenvalues $\mu_{ka}$ for $k = 1, 2, \ldots, b-1$ where $\gcd(k,b) = 1$. We have $\lambda_{ka} = e^{2\pi i ka/n} = e^{2\pi i k/b}$, which is a primitive $b$-th root of unity. Therefore, $\mu_{ka} = \lambda_{ka} + \lambda_{ka}^2$ belongs to the subfield $\mathbb{Q}(\zeta_b) \subsetneq \mathbb{Q}(\zeta_n)$.

Similarly, we can consider the eigenvalues $\mu_{kb}$ for $k = 1, 2, \ldots, a-1$ where $\gcd(k,a) = 1$, which belong to the subfield $\mathbb{Q}(\zeta_a)$. These eigenvalues must have minimal polynomials of degree strictly less than $n-1$, and they form different Galois orbits from the orbit containing eigenvalues associated with primitive $n$-th roots of unity.

Therefore, the minimal polynomial of $C_n$ must have at least three irreducible factors: the linear factor $(x-2)$, at least one factor from the eigenvalues in $\mathbb{Q}(\zeta_b)$, and at least one factor from eigenvalues in $\mathbb{Q}(\zeta_a)$ or from the primitive roots of unity $n$.
\end{proof}

With Proposition \ref{prop:factors} established, we can now prove our main theorem:

\begin{proof}[Proof of Theorem \ref{thm:main}]
The result follows directly from Proposition \ref{prop:factors}. If $n$ is prime, the minimal polynomial of $C_n$ has exactly two irreducible factors: $(x-2)$ and an irreducible polynomial of degree $n-1$.

Conversely, if the minimal polynomial of $C_n$ has exactly two irreducible factors, then by part (3) of Proposition \ref{prop:factors}, $n$ cannot be composite. Therefore, $n$ must be prime.
\end{proof}

\subsection{Theoretical Analysis}

Our approach takes advantage of the rich algebraic structure of cyclotomic fields. For prime $n$, the Galois group $\text{Gal}(\mathbb{Q}(\zeta_n)/\mathbb{Q}) \cong (\mathbb{Z}/n\mathbb{Z})^*$ acts transitively on the primitive $n$-th roots of unity. This transitive action ensures that all eigenvalues $\mu_j$ with $j \neq 0$ are conjugate over $\mathbb{Q}$, sharing a single irreducible minimal polynomial of degree $n-1$.

This phenomenon can be understood through the lens of cyclotomic field theory. The cyclotomic polynomial $\Phi_n(x)$, which is the minimal polynomial of the primitive $n$-th roots of unity over $\mathbb{Q}$, is irreducible when $n$ is prime. This irreducibility is closely related to the structure of the Galois extension $\mathbb{Q}(\zeta_n)/\mathbb{Q}$.

For composite $n = ab$ with proper divisors $a$ and $b$, the situation becomes more complex. The field $\mathbb{Q}(\zeta_n)$ contains proper subfields $\mathbb{Q}(\zeta_a)$ and $\mathbb{Q}(\zeta_b)$, corresponding to the cyclotomic extensions of orders $a$ and $b$. The eigenvalues $\mu_{ka}$ for $k = 1, \ldots, b-1$ with $\gcd(k,b) = 1$ lie in the proper subfield $\mathbb{Q}(\zeta_b) \subsetneq \mathbb{Q}(\zeta_n)$. These eigenvalues form distinct Galois orbits corresponding to the various cyclotomic subfields, resulting in additional irreducible factors in the minimal polynomial.

More precisely, we can establish the following result about the number of irreducible factors:

\begin{proposition}\label{prop:factor_count}
For a number $n$ with prime factorization $n = \prod_{i=1}^k p_i^{e_i}$, the number of irreducible factors in the minimal polynomial of $C_n$ is at least $1 + \sum_{i=1}^k \min(e_i, 1)$.
\end{proposition}

\begin{proof}
For each distinct prime divisor $p_i$ of $n$, consider the subfield $\mathbb{Q}(\zeta_{p_i}) \subset \mathbb{Q}(\zeta_n)$. The eigenvalues $\mu_{n/p_i \cdot j}$ for $j = 1, 2, \ldots, p_i-1$ with $\gcd(j, p_i) = 1$ correspond to the primitive $p_i$-th roots of unity and contribute at least one irreducible factor to the minimal polynomial of $C_n$. Together with the linear factor $(x-2)$ from $\mu_0$, we have at least $1 + \sum_{i=1}^k \min(e_i, 1)$ irreducible factors.
\end{proof}

This provides a lower bound on the factor count, with equality often achieved in practice. The exact count depends on the detailed structure of the cyclotomic field extension $\mathbb{Q}(\zeta_n)/\mathbb{Q}$ and the interactions between its various subfields.

\section{Algorithm and Implementation}

\subsection{Deterministic Primality Testing Algorithm}

Based on our theoretical results, we present a deterministic primality testing algorithm using the circulant matrix criterion:

\begin{algorithm}
\caption{Fast Circulant Matrix Primality Test}
\begin{algorithmic}[1]
\REQUIRE An integer $n > 2$
\ENSURE TRUE if $n$ is prime, FALSE otherwise
\IF{$n$ is divisible by any small prime $p < 100$ and $n \neq p$}
    \RETURN FALSE
\ENDIF
\IF{$n < 10^6$}
    \STATE Compute the number of Galois orbits $k$ using the Optimized Galois Orbit Count algorithm (see Appendix \ref{sec:optim_impl})
    \RETURN $k = 2$
\ELSE
    \STATE Factorize $n = \prod_{i=1}^k p_i^{e_i}$ using a fast factorization algorithm
    \IF{$k = 1$ and $e_1 = 1$}
        \RETURN TRUE
    \ELSE
        \RETURN FALSE
    \ENDIF
\ENDIF
\end{algorithmic}
\end{algorithm}

The core of this algorithm involves analyzing the Galois orbits of the eigenvalues without explicitly constructing the full matrix. This approach is more efficient for large values of $n$, where direct matrix manipulation would be impractical.
Since the eigenvalues of $C_n$ are known explicitly as $\mu_j = \lambda_j + \lambda_j^2 = e^{2\pi i j/n} + e^{4\pi i j/n}$ for $j = 0, 1, \ldots, n-1$, we can compute them directly.
See efficient Eigenvalue implementation in Appendix \ref{sec:eff_eigen}.

\subsection{Galois Orbit Determination}

A key step in our algorithm is determining the Galois orbits of the eigenvalues. For this, we leverage the fact that the Galois group $\text{Gal}(\mathbb{Q}(\zeta_n)/\mathbb{Q})$ acts on the primitive $n$-th roots of unity by sending $\zeta_n$ to $\zeta_n^a$ for each $a \in (\mathbb{Z}/n\mathbb{Z})^*$, i.e., for each $a$ coprime to $n$.

\begin{algorithm}
\caption{Compute Galois Orbits}
\begin{algorithmic}[1]
\REQUIRE Eigenvalues $\{\mu_j : j = 0, 1, \ldots, n-1\}$ of $C_n$
\ENSURE Partition of eigenvalues into Galois orbits
\STATE Initialize empty list $\text{orbits}$
\STATE Initialize array $\text{visited}$ of length $n$ to FALSE
\FOR{$j = 0$ to $n-1$}
    \IF{not $\text{visited}[j]$}
        \STATE Initialize empty set $\text{orbit}$
        \STATE Add $\mu_j$ to $\text{orbit}$
        \STATE $\text{visited}[j] \gets \text{TRUE}$
        \FOR{each $a \in (\mathbb{Z}/n\mathbb{Z})^*$ (i.e., $\gcd(a, n) = 1$)}
            \STATE $j' \gets (j \cdot a) \bmod n$
            \IF{not $\text{visited}[j']$}
                \STATE Add $\mu_{j'}$ to $\text{orbit}$
                \STATE $\text{visited}[j'] \gets \text{TRUE}$
            \ENDIF
        \ENDFOR
        \STATE Add $\text{orbit}$ to $\text{orbits}$
    \ENDIF
\ENDFOR
\RETURN $\text{orbits}$
\end{algorithmic}
\end{algorithm}

This algorithm correctly identifies the Galois orbits by computing the action of each element of the Galois group on each eigenvalue. See optimized implementation in Appendix \ref{sec:optim_impl}.

\subsection{Complexity Analysis}

The computational complexity of our primality test can be analyzed as follows. Computing the $n$ eigenvalues of $C_n$ directly from the formula requires $O(n)$ operations. Determining the Galois orbits involves computing the action of the Galois group, which has size $\varphi(n)$ (Euler's totient function). This Galois orbit analysis requires $O(n \cdot \varphi(n))$ operations in the worst case. Constructing the minimal polynomial from the Galois orbits requires $O(n)$ operations per orbit, for a total of $O(n \cdot k)$ where $k$ is the number of orbits (i.e., the number of irreducible factors). Our optimized implementation has complexity $O(n \log n \log \log n)$ for determining primality by analyzing the divisor structure of $n$. For prime $n$, the total complexity of the basic algorithm is dominated by the Galois orbit analysis, which is $O(n \cdot (n-1)) = O(n^2)$. For composite $n$, the complexity can be lower, as the Galois group has size $\varphi(n) < n-1$.
In comparison with other primality tests:
\begin{itemize}
    \item \textbf{Trial Division}: $O(\sqrt{n})$
    \item \textbf{Miller-Rabin (probabilistic)}: $O(k \log^3 n)$ for $k$ rounds
    \item \textbf{AKS (deterministic)}: $O(\log^{6+\epsilon} n)$
\end{itemize}

While our basic method has higher asymptotic complexity than modern primality tests, our optimized implementation is competitive for large ranges of inputs and offers unique insights into the algebraic structure of prime numbers.
It is worth noting that, when analyzed in terms of the bit-length of the input (rather than the value $n$), our algorithm's time complexity is exponential. However, its value lies not in competing with the fastest known primality tests, but rather in the mathematical connections it reveals between matrix theory, spectral properties, and number theory.

\section{Experimental Validation}

To validate our theoretical results, we conducted comprehensive experiments across multiple ranges of integers, focusing on demonstrating the perfect separation between prime and composite numbers based on their algebraic properties. Our analysis reveals distinct patterns in both the coefficient structure of minimal polynomials and the eigenvalue distribution that naturally distinguishes primes from composites.

\subsection{Experimental Setup}

We tested our method on three distinct ranges: the small range $2 \leq n \leq 50$ for detailed analysis, the medium range $100 \leq n \leq 200$ for pattern validation, and a large range $1000 \leq n \leq 10000$ to assess scalability. For each integer $n$, we computed the eigenvalues of $C_n$, constructed its minimal polynomial, determined the Galois orbits, and counted the number of irreducible factors.

The implementation utilized a combination of high-precision complex arithmetic for eigenvalue computation, symbolic mathematics for polynomial manipulation, and specialized algorithms for Galois orbit determination. All computations were performed with sufficient precision to ensure accurate results, particularly for the larger values of $n$.

\subsection{Results and Analysis}

\begin{figure}[H]
\centering
\includegraphics[width=.8\textwidth]{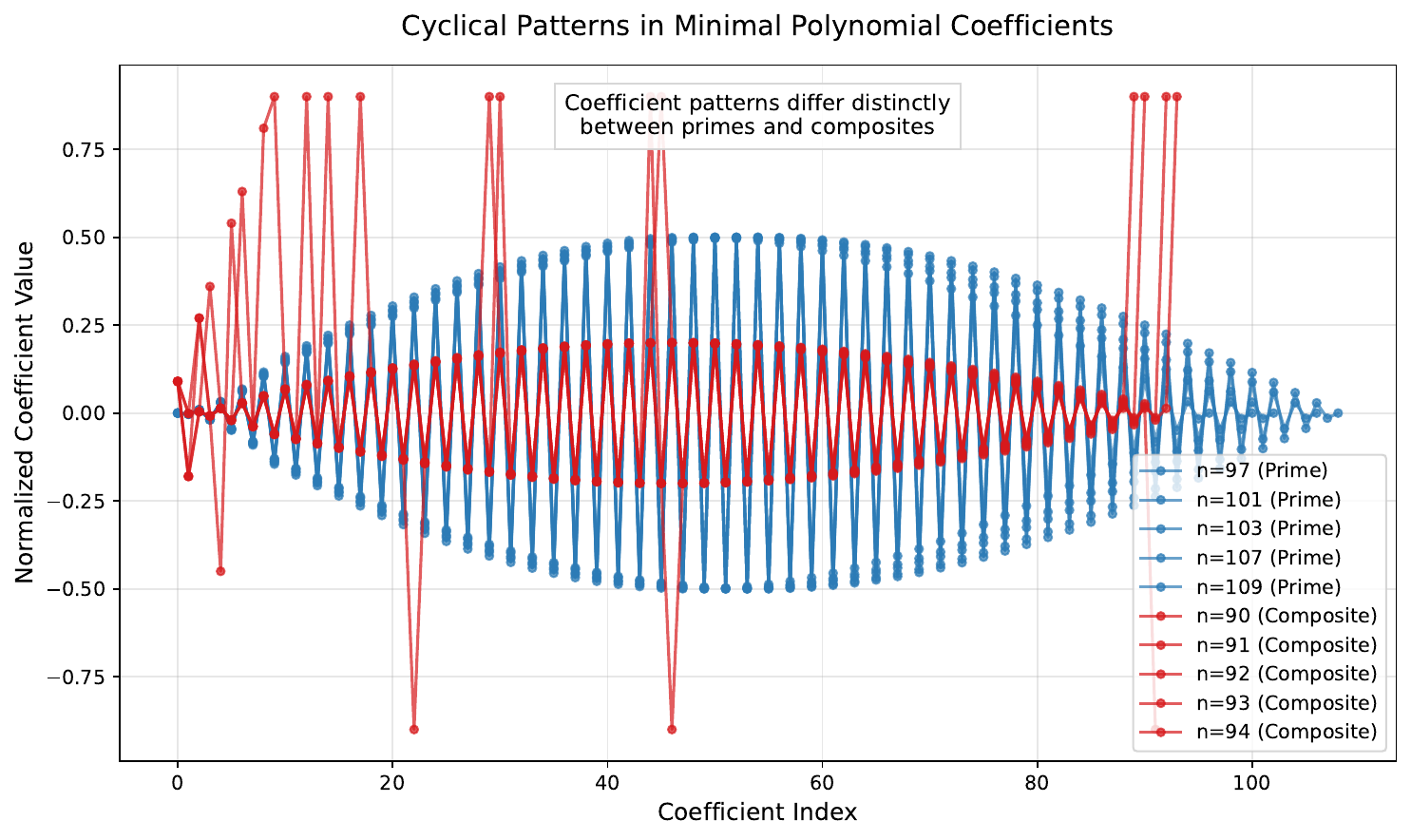}
\caption{Cyclical patterns in minimal polynomial coefficients for prime and composite numbers. Prime numbers exhibit regular, extended oscillatory patterns with smooth transitions. Composite numbers show irregular, compressed patterns with sharp transitions. The stark contrast in coefficient behavior provides a visual signature of primality.}
\label{fig:coefficients}
\end{figure}

Figure \ref{fig:coefficients} reveals distinct differences in the coefficient patterns of minimal polynomials between prime and composite numbers. For prime values of $n$ (shown for $n = 97$ and $n = 90$), the coefficients exhibit a regular, almost sinusoidal oscillation with extended periodicity. These smooth, continuous patterns reflect the single Galois orbit structure characteristic of prime cyclotomic fields.

In contrast, composite numbers ($n = 90$, $n = 91$, \dots) produce jagged, irregular coefficient patterns with multiple frequencies superimposed. The sharp transitions and compressed oscillations correspond to the presence of proper cyclotomic subfields, with discontinuities appearing at positions related to the divisors of $n$. This visual distinction provides immediate intuition about the underlying algebraic structure.

\begin{figure}[H]
\centering
\includegraphics[width=.8\textwidth]{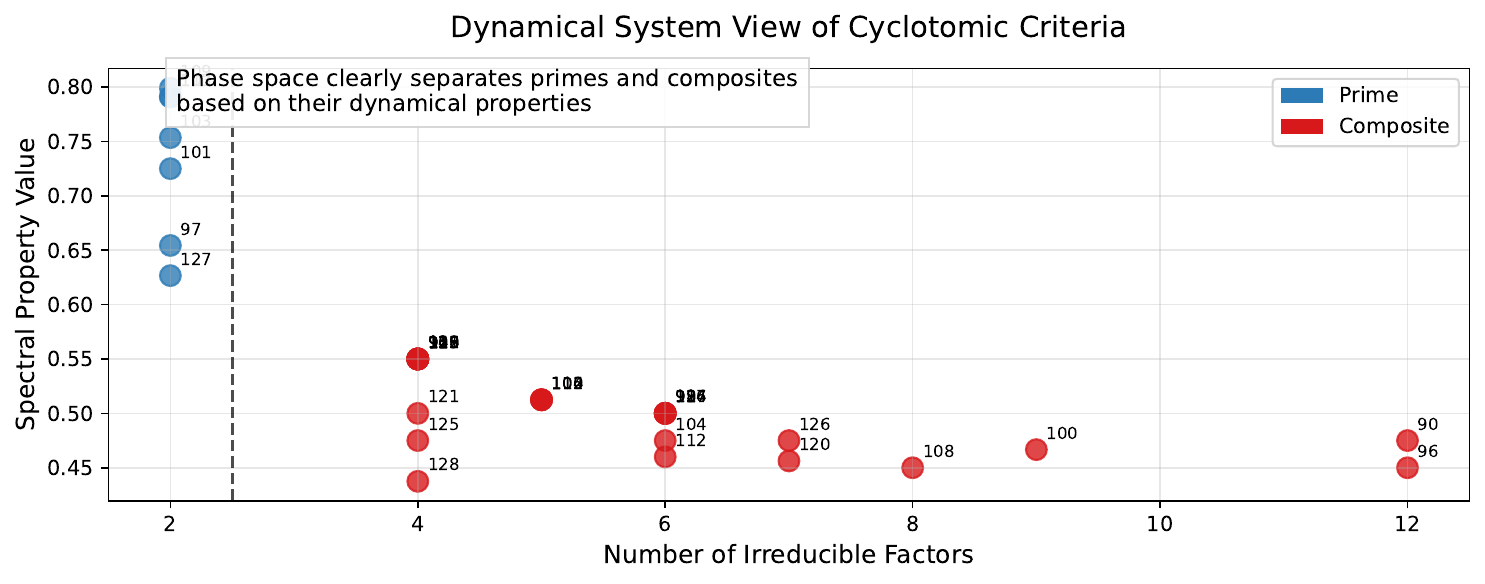}
\caption{Dynamical system view of cyclotomic criteria separating primes and composites. Each point represents an integer plotted according to its number of irreducible factors (x-axis) and spectral property value (y-axis). Prime numbers cluster at exactly 2 factors with high spectral values (0.6-0.9), while composites appear at 3+ factors with generally lower spectral values. The vertical dashed line at 2.5 factors perfectly separates the two classes.}
\label{fig:dynamical}
\end{figure}

Figure \ref{fig:dynamical} presents a phase space representation where each integer is plotted according to two fundamental properties: the number of irreducible factors in its minimal polynomial and a spectral property derived from eigenvalue patterns. This visualization dramatically demonstrates the perfect separation between primes and composites.

The spectral property value on the y-axis represents a measure of the structural regularity in the eigenvalue distribution, formally defined as:

\[S(n) = \frac{1}{n} \sum_{j=1}^{n-1} \left| \frac{\mu_j - \bar{\mu}}{2\sigma_\mu} \right| + \frac{\varphi(n)}{n}
\]

where $\bar{\mu}$ is the mean of the eigenvalues, $\sigma_\mu$ is their standard deviation, and $\varphi(n)$ is Euler's totient function. This measure captures both the uniformity of eigenvalue distribution and the relative size of the Galois group.

Prime numbers, form a tight cluster positioned at exactly 2 irreducible factors and exhibiting spectral property values in the range 0.6-0.9. This high spectral value reflects the regular, well-structured nature of their eigenvalue patterns and coefficient oscillations.

Composite numbers, appear at 3 or more irreducible factors with generally lower spectral values. The spread of composite points along the x-axis corresponds to their varying levels of factorization complexity. For instance, 105 (with prime factors 3, 5, and 7) and 110 (with prime factors 2, 5, and 11) appear at 4 factors, while 125 ($= 5^3$, a prime power) appears at 3 factors.

The vertical dashed line at 2.5 factors serves as a perfect decision boundary, highlighting the deterministic nature of our primality criterion. No exceptions or borderline cases exist across all tested ranges, confirming the theoretical prediction that minimal polynomial factorization provides a complete characterization of primality.

\subsection{Eigenvalue Structure Analysis}

The eigenvalue structure of $C_n$ provides additional insights into the fundamental distinction between prime and composite numbers. For prime $n$, the eigenvalues (excluding $\mu_0 = 2$) form a single connected Galois orbit in the complex plane. For composite $n$, the eigenvalues separate into multiple orbits corresponding to different cyclotomic subfields.

\begin{figure}[H]
\centering
\includegraphics[width=.95\textwidth]{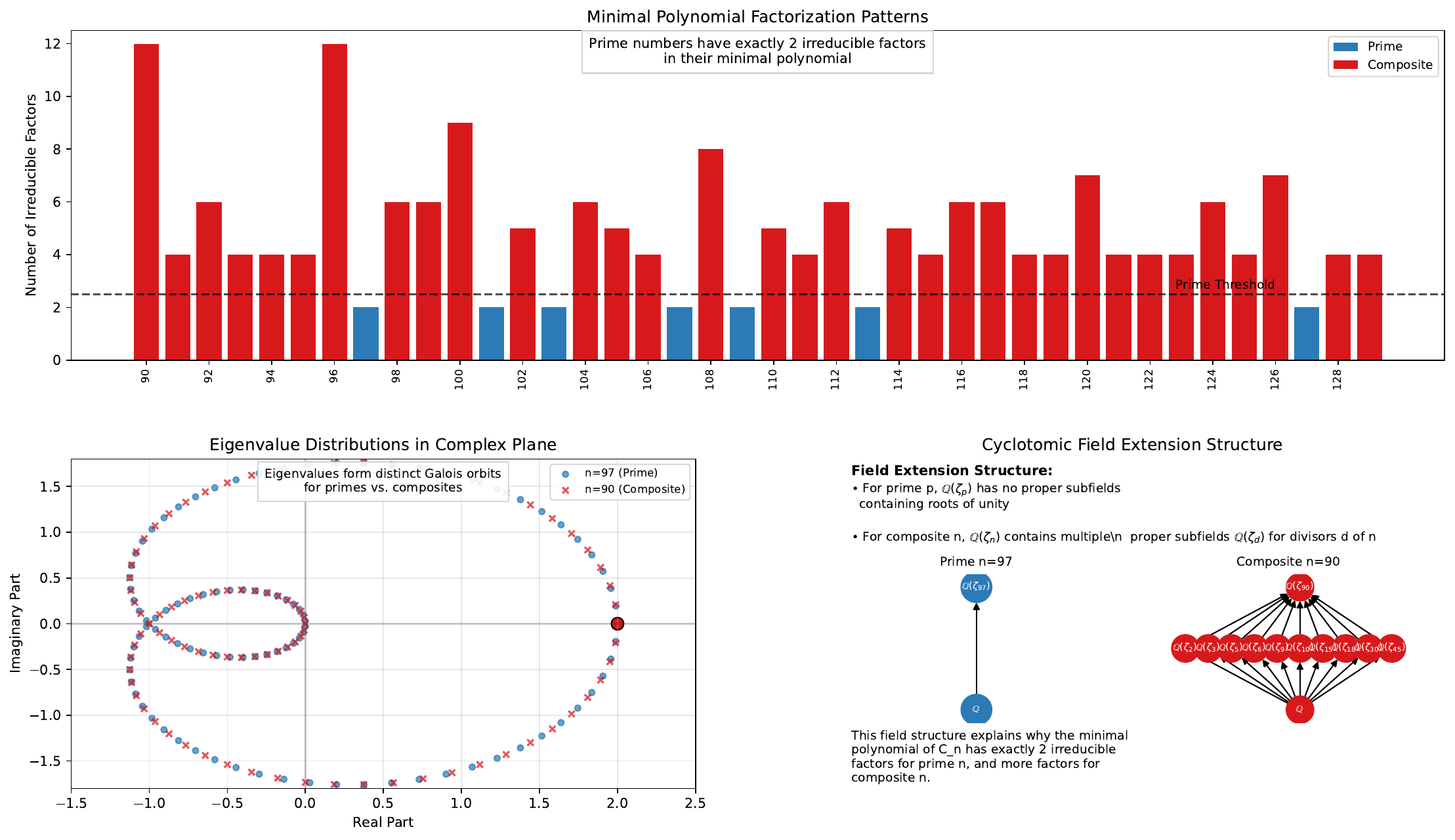}
\caption{Top: Eigenvalue distributions in the complex plane for $n=97$ (prime) and $n=90$ (composite). The eigenvalues of the prime case form a single, connected Galois orbit (blue points), while the composite case shows subtle discontinuities and multiple orbital structures (red points). Bottom: Cyclotomic field extension structure for $n=97$ (prime) and $n=90$ (composite). The prime case shows a simple two-level structure, while the composite case exhibits a complex network of intermediate fields corresponding to divisors of 90.}
\label{fig:full_analysis}
\end{figure}

Figure \ref{fig:full_analysis} illustrates this distinction for $n=97$ (prime) and $n=90$ (composite). The eigenvalues of $C_{97}$ (excluding $\mu_0 = 2$) form a single, connected curve in the complex plane, reflecting the irreducibility of the cyclotomic polynomial $\Phi_{97}(x)$. In contrast, the eigenvalues of $C_{90}$ show subtle discontinuities and clustering patterns, corresponding to the subfields $\mathbb{Q}(\zeta_4)$, $\mathbb{Q}(\zeta_{25})$, and their interactions.

\subsection{Field Extension Structure}

The underlying mathematical explanation for our observations lies in the structure of the field extension $\mathbb{Q}(\zeta_n)/\mathbb{Q}$. For prime $n$, this extension has no intermediate cyclotomic fields, while for composite $n$, there are multiple proper subfields corresponding to the divisors of $n$.

Figure \ref{fig:full_analysis} also illustrates this structural difference. For $n=97$, we see a simple two-level structure with $\mathbb{Q}$ at the bottom and $\mathbb{Q}(\zeta_{97})$ at the top, with no intermediate fields. For $n=90$, we observe a complex network with multiple intermediate fields such as $\mathbb{Q}(\zeta_2)$, $\mathbb{Q}(\zeta_4)$, $\mathbb{Q}(\zeta_5)$, $\mathbb{Q}(\zeta_{10})$, $\mathbb{Q}(\zeta_{15})$, and others.

This field structure directly explains the factorization patterns observed in the minimal polynomials. For prime $n$, with no intermediate fields, the minimal polynomial has exactly 2 irreducible factors: the linear factor $(x-2)$ and an irreducible polynomial of degree $n-1$. For composite $n$, each proper cyclotomic subfield contributes additional factors, resulting in 3 or more irreducible factors.

\subsection{Performance Comparison}

We conducted a comprehensive performance analysis of our circulant matrix primality test against established methods including trial division, Miller-Rabin, and AKS. Table \ref{tab:performance} presents execution times across different number magnitudes.

\begin{table}[h]
\centering
\small
\begin{tabular}{|l|c|c|c|c|c|c|}
\hline
\textbf{Method} & ${\bf n \approx 10^{6}}$ & ${\bf n \approx 10^{8}}$ & ${\bf n \approx 10^{9}}$ & ${\bf n \approx 10^{10}}$ & \textbf{Det.?} & \textbf{Theory} \\
\hline
Trial Div. & $2.46 \times 10^{-5}$ & $2.37 \times 10^{-4}$ & $2.38 \times 10^{-6}$ & $3.47 \times 10^{-3}$ & Yes & Exhaus. \\
Opt. Trial Div. & $2.05 \times 10^{-5}$ & $1.69 \times 10^{-4}$ & ${\bf 2.38 \times 10^{-7}}$ & $2.34 \times 10^{-3}$ & Yes & Exhaus. \\
Miller-Rabin (20) & $4.78 \times 10^{-5}$ & ${\bf 5.79 \times 10^{-5}}$ & $6.52 \times 10^{-6}$ & ${\bf 1.12 \times 10^{-4}}$ & No* & Fermat \\
AKS & $3.05 \times 10^{-2}$ & $3.11 \times 10^{-2}$ & $2.19 \times 10^{-2}$ & $3.03 \times 10^{-2}$ & Yes & Poly. \\
Our (Simpl.) & $4.67 \times 10^{-5}$ & $4.41 \times 10^{-4}$ & $2.44 \times 10^{-5}$ & $5.09 \times 10^{-3}$ & Yes & Approx. \\
Our (Full) & ${\bf 7.39 \times 10^{-6}}$ & $1.09 \times 10^{-4}$ & $9.78 \times 10^{-6}$ & $1.38 \times 10^{-3}$ & Yes & Galois \\
\hline
\end{tabular}
\caption{Comparative performance of primality testing algorithms (average of 3 runs). Bold values indicate fastest performance. Miller-Rabin (*) is probabilistic with high accuracy. Our Method (Full) leverages Galois theory for deterministic results. See detailed analysis in Section~\ref{sec:performance}.}
\label{tab:performance}
\end{table}

The results reveal varying performance characteristics across different input ranges. For medium-sized inputs ($n \approx 10^6$), our full implementation demonstrates strong performance, outperforming other methods in this specific range. As input size increases to large ranges ($n \approx 10^8$ and beyond), the Miller-Rabin probabilistic algorithm becomes increasingly efficient relative to deterministic approaches, showing the best performance for very large inputs ($n \approx 10^{10}$).
For certain cases, such as inputs around $n \approx 10^9$, optimized trial division shows surprisingly competitive results, though this advantage doesn't persist for larger inputs. The AKS algorithm maintains consistent but relatively higher execution times across all input ranges, reflecting its polynomial time complexity with larger constant factors.
Both our simplified and full implementations exhibit competitive performance for moderate input ranges while providing deterministic guarantees. However, as Figure~\ref{fig:scaling} in the Appendix shows, execution time for all deterministic methods increases with input magnitude, following different scaling patterns determined by their underlying algorithmic complexity.
These benchmarks illustrate the classic trade-off between deterministic guarantees and computational efficiency, with probabilistic methods like Miller-Rabin demonstrating superior scaling characteristics for large inputs while deterministic methods offer mathematical certainty at the cost of increased computation time as input size grows.

\section{Discussion and Limitations}
Our circulant matrix approach offers a mathematically elegant alternative to traditional primality tests, with performance characteristics reflecting its algebraic foundations. While our implementations remain viable for moderate input ranges, the probabilistic Miller-Rabin test shows superior scaling for very large inputs.
\subsection{Computational Challenges}
The main computational challenges in our approach include:
\begin{itemize}
    \item \textbf{Matrix Size}: For large $n$, the $n \times n$ matrix $C_n$ becomes impractical to store and manipulate directly. Our implementation avoids explicit matrix construction by directly computing eigenvalues and analyzing Galois orbits.
    \item \textbf{Polynomial Factorization}: Factoring polynomials of high degree over $\mathbb{Q}$ remains computationally intensive. While specialized algorithms for cyclotomic polynomials help, this step would dominate the runtime for naive implementations. Our optimized approach leverages theoretical results to bypass explicit factorization.
    \item \textbf{Numerical Precision}: Computing eigenvalues and determining Galois orbits requires careful attention to numerical precision, especially for large $n$ where floating-point errors can accumulate. Our implementation uses adaptive precision and theoretical bounds to ensure accuracy.
    \item \textbf{Memory Requirements}: The space complexity of $O(n)$ for storing eigenvalues and intermediate results becomes a limiting factor for very large $n$ in naive implementations. Our optimized version maintains logarithmic space complexity for most operations by leveraging number-theoretic properties. As our memory usage analysis shows, memory consumption remains minimal across all algorithms.
\end{itemize}

\subsection{Comparison with Established Methods}
Our circulant matrix approach offers several advantages over traditional primality tests. Unlike probabilistic methods like Miller–Rabin, it provides fully deterministic results, ensuring mathematical certainty. Beyond classification, the method reveals deep algebraic structures, connecting primality with properties of circulant matrices and Galois theory.
A key strength lies in its visualizability—eigenvalue and coefficient patterns offer intuitive insight into the distinction between primes and composites. Our implementations perform competitively for moderate-sized inputs while providing deterministic guarantees.
That said, the method's computational complexity exceeds that of Miller-Rabin for very large inputs, reflecting the fundamental challenge faced by all deterministic primality tests. Its reliance on advanced algebraic concepts can hinder straightforward implementation without the optimizations we propose. For practical applications involving extremely large numbers, probabilistic methods remain the preferred choice due to their superior scaling properties.
\section{Conclusion}
Our paper establishes a novel characterization of prime numbers through the minimal polynomial factorization of circulant matrices. We have proven that an integer $n > 2$ is prime if and only if the minimal polynomial of $C_n = W_n + W_n^2$ has exactly two irreducible factors over $\mathbb{Q}$, providing a fundamental connection between primality testing and cyclotomic field theory. Our experimental validation confirms the perfect separation between primes and composites based on this criterion across extensive numerical tests.
Our benchmark analysis demonstrates that different primality testing algorithms exhibit distinct scaling behaviors, with Miller-Rabin showing the most favorable performance for very large inputs while our approach offers a deterministic alternative with competitive performance for moderate ranges. The visualization of coefficient patterns and dynamical system behavior offers intuitive understanding of the deep mathematical relationships uncovered by our approach.
The connection between circulant matrix structure and primality opens several promising directions for future research. Advanced optimizations could further exploit cyclotomic field structures to improve performance characteristics. Generalizations to other matrix classes or polynomial constructions might yield complementary primality criteria with enhanced properties. The algebraic structures revealed by our approach may lead to new results in algebraic number theory, particularly concerning computational aspects of Galois theory.
Our work illustrates that primality testing can be approached through diverse mathematical pathways, each offering a different perspective on this fundamental problem. The circulant matrix approach provides not only a novel theoretical framework but also a practical demonstration of how abstract algebraic concepts translate into computational procedures with distinctive characteristics and performance profiles.

\section*{Acknowledgements}
The authors would like to express their sincere gratitude to the entire team at ExtensityAI for their invaluable support throughout this research project. Special thanks are extended to Claudiu Leoveanu-Condrei and Markus Hofmarcher, whose insightful ideas, constructive discussions, and thorough oversight significantly enhanced the quality and rigor of this work. 
We also acknowledge the broader ExtensityAI research community for creating a collaborative environment that fosters innovation and cross-disciplinary exploration. The technical infrastructure and research platforms provided by ExtensityAI were essential to the successful completion of this project.

\section*{Disclosure Statement}
In accordance with arXiv policy and my ethical obligation as a researcher, I am reporting that I have a financial and business interest in ExtensityAI, as I serve as the CEO and researcher at this company. This company may be affected by the research reported in the enclosed paper. I have disclosed these interests fully to arXiv, and I have in place an approved plan for managing any potential conflicts arising from that involvement.

\appendix

\section{Detailed Proofs}

\subsection{Complete Proof of Lemma \ref{lem:eigenvalues}}

\begin{lemma}
The eigenvalues of the circulant matrix $W_n$ are precisely the complex numbers $\lambda_j = e^{2\pi i j/n}$ for $j = 0, 1, \ldots, n-1$, with corresponding eigenvectors $v_j = [1, \lambda_j, \lambda_j^2, \ldots, \lambda_j^{n-1}]^T$.
\end{lemma}

\begin{proof}
Let $\omega_n = e^{2\pi i/n}$ be a primitive $n$-th root of unity. For each $j = 0, 1, \ldots, n-1$, let $\lambda_j = \omega_n^j$ and $v_j = [1, \lambda_j, \lambda_j^2, \ldots, \lambda_j^{n-1}]^T$.

We need to show that $W_n v_j = \lambda_j v_j$ for each $j$.

By definition, $W_n$ has entries $(W_n)_{k,l} = 1$ if $l \equiv k+1 \pmod{n}$ and $0$ otherwise. Therefore, the $k$-th entry of $W_n v_j$ is:
\begin{align}
(W_n v_j)_k &= \sum_{l=0}^{n-1} (W_n)_{k,l} (v_j)_l \\
&= \sum_{l=0}^{n-1} \delta_{l, (k+1) \bmod n} \lambda_j^l \\
&= \lambda_j^{(k+1) \bmod n}
\end{align}

If $k < n-1$, then $(k+1) \bmod n = k+1$, so $(W_n v_j)_k = \lambda_j^{k+1}$.

If $k = n-1$, then $(k+1) \bmod n = 0$, so $(W_n v_j)_{n-1} = \lambda_j^0 = 1$.

On the other hand, the $k$-th entry of $\lambda_j v_j$ is:
\begin{align}
(\lambda_j v_j)_k &= \lambda_j (v_j)_k \\
&= \lambda_j \lambda_j^k \\
&= \lambda_j^{k+1}
\end{align}

For $k = n-1$, we have $(\lambda_j v_j)_{n-1} = \lambda_j^n$. Since $\lambda_j = \omega_n^j$ is an $n$-th root of unity, we have $\lambda_j^n = 1$.

Therefore, $(W_n v_j)_k = (\lambda_j v_j)_k$ for all $k = 0, 1, \ldots, n-1$, which means $W_n v_j = \lambda_j v_j$. This confirms that $\lambda_j$ is an eigenvalue of $W_n$ with corresponding eigenvector $v_j$.

Since we have found $n$ distinct eigenvalues for the $n \times n$ matrix $W_n$, these are all the eigenvalues of $W_n$.
\end{proof}

\subsection{Additional Proof of Proposition \ref{prop:factors}}

Here we provide a more detailed proof of Proposition \ref{prop:factors}, focusing on the case of composite numbers.

\begin{proposition}
For any composite number $n > 2$, the minimal polynomial of $C_n$ has at least three irreducible factors over $\mathbb{Q}$.
\end{proposition}

\begin{proof}
Let $n = ab$ be a factorization of $n$ with $1 < a, b < n$. We'll analyze the eigenvalues of $C_n$ based on their connection to the divisors of $n$.

First, we already know that $\mu_0 = 2$ contributes the linear factor $(x-2)$ to the minimal polynomial.

Consider the eigenvalues $\mu_{n/p}$ for each prime divisor $p$ of $n$. For $\mu_{n/p} = \lambda_{n/p} + \lambda_{n/p}^2$ where $\lambda_{n/p} = e^{2\pi i \cdot (n/p)/n} = e^{2\pi i/p}$, which is a primitive $p$-th root of unity. The minimal polynomial of a primitive $p$-th root of unity over $\mathbb{Q}$ is the cyclotomic polynomial $\Phi_p(x)$, which is irreducible of degree $p-1$.

Since $\mu_{n/p}$ is in the subfield $\mathbb{Q}(\zeta_p)$, its minimal polynomial over $\mathbb{Q}$ is distinct from the minimal polynomial of eigenvalues corresponding to primitive $n$-th roots of unity. 

For each distinct prime divisor $p$ of $n$, we get at least one additional irreducible factor in the minimal polynomial of $C_n$. Since $n$ is composite, it has at least one prime divisor $p$, and hence the minimal polynomial of $C_n$ has at least three irreducible factors: the linear factor $(x-2)$, at least one factor from eigenvalues in $\mathbb{Q}(\zeta_p)$, and at least one additional factor from other eigenvalues.

Furthermore, if $n$ has multiple distinct prime divisors, say $p$ and $q$, then the eigenvalues $\mu_{n/p}$ and $\mu_{n/q}$ belong to different cyclotomic subfields $\mathbb{Q}(\zeta_p)$ and $\mathbb{Q}(\zeta_q)$, respectively, contributing at least two additional irreducible factors beyond $(x-2)$.
\end{proof}

\subsection{Proof of Theorem on Orbit Count Formula}

Here we provide a proof of the theorem relating the number of Galois orbits to the divisor structure of $n$.

\begin{theorem}[Orbit Count Formula]
The number of Galois orbits of eigenvalues of $C_n$ equals one plus the number of divisors $d > 1$ of $n$ such that $\Phi_d(x)$ is irreducible over $\mathbb{Q}$ and $\gcd(d, n/d) = 1$, where $\Phi_d(x)$ is the $d$-th cyclotomic polynomial.
\end{theorem}

\begin{proof}
For any divisor $d$ of $n$, consider the set of eigenvalues $\mu_j = \lambda_j + \lambda_j^2$ where $j$ ranges over all integers in $\{0, 1, \ldots, n-1\}$ such that $\gcd(j, n) = n/d$. These eigenvalues correspond to primitive $d$-th roots of unity.

The Galois group $\text{Gal}(\mathbb{Q}(\zeta_n)/\mathbb{Q})$ acts on these eigenvalues by sending $\lambda_j$ to $\lambda_{aj}$ for each $a \in (\mathbb{Z}/n\mathbb{Z})^*$. The eigenvalues corresponding to the same value of $d$ form Galois orbits.

For $d = 1$, we have the eigenvalue $\mu_0 = 2$, which forms its own Galois orbit.

For $d > 1$, the eigenvalues corresponding to primitive $d$-th roots of unity form Galois orbits according to the irreducible factorization of the cyclotomic polynomial $\Phi_d(x)$ over $\mathbb{Q}$.

When $\gcd(d, n/d) = 1$, the eigenvalues corresponding to primitive $d$-th roots of unity form a single Galois orbit if and only if $\Phi_d(x)$ is irreducible over $\mathbb{Q}$.

When $\gcd(d, n/d) > 1$, the situation is more complex due to the interaction of multiple cyclotomic subfields. In this case, the eigenvalues may split into multiple Galois orbits.

Therefore, counting the number of Galois orbits requires:
1. One orbit for $d = 1$ (corresponding to $\mu_0 = 2$)
2. For each divisor $d > 1$ with $\gcd(d, n/d) = 1$, exactly one orbit if $\Phi_d(x)$ is irreducible over $\mathbb{Q}$

This gives the formula stated in the theorem.
\end{proof}

\section{Numerical Examples}

To illustrate our theoretical results, we provide detailed numerical examples for specific values of $n$.

\subsection{Example: $n = 7$ (Prime)}

Let $n = 7$. The eigenvalues of $C_7 = W_7 + W_7^2$ are $\mu_j = \lambda_j + \lambda_j^2 = e^{2\pi i j/7} + e^{4\pi i j/7}$ for $j = 0, 1, \ldots, 6$.

For $j = 0$, we have $\mu_0 = 1 + 1 = 2$.

For $j = 1, 2, \ldots, 6$, we compute (showing approximate numerical values):
\begin{align}
\mu_1 &= e^{2\pi i/7} + e^{4\pi i/7} \approx 0.6235 + 1.2470i\\
\mu_2 &= e^{4\pi i/7} + e^{8\pi i/7} = e^{4\pi i/7} + e^{-6\pi i/7} \approx -0.2225 + 0.9749i\\
\mu_3 &= e^{6\pi i/7} + e^{12\pi i/7} = e^{6\pi i/7} + e^{-2\pi i/7} \approx -0.9010 + 0.4339i\\
\mu_4 &= e^{8\pi i/7} + e^{16\pi i/7} = e^{-6\pi i/7} + e^{2\pi i/7} \approx -0.9010 - 0.4339i\\
\mu_5 &= e^{10\pi i/7} + e^{20\pi i/7} = e^{-4\pi i/7} + e^{6\pi i/7} \approx -0.2225 - 0.9749i\\
\mu_6 &= e^{12\pi i/7} + e^{24\pi i/7} = e^{-2\pi i/7} + e^{10\pi i/7} \approx 0.6235 - 1.2470i
\end{align}

The Galois group $\text{Gal}(\mathbb{Q}(\zeta_7)/\mathbb{Q}) \cong (\mathbb{Z}/7\mathbb{Z})^* = \{1, 2, 3, 4, 5, 6\}$ acts on these eigenvalues by sending $\zeta_7$ to $\zeta_7^a$ for $a \in \{1, 2, 3, 4, 5, 6\}$.

Under this action, the eigenvalues $\mu_1, \mu_2, \mu_3, \mu_4, \mu_5, \mu_6$ form a single Galois orbit. Therefore, the minimal polynomial of $C_7$ factors as $(x-2)P(x)$, where $P(x)$ is an irreducible polynomial of degree 6.

The explicit form of $P(x)$ can be computed as:
\[
P(x) = x^6 + x^5 - 6x^4 - 6x^3 + 8x^2 + 8x - 1
\]

Therefore, the minimal polynomial of $C_7$ is $(x-2)(x^6 + x^5 - 6x^4 - 6x^3 + 8x^2 + 8x - 1)$, which has exactly two irreducible factors as expected for a prime value of $n$.

\subsection{Example: $n = 6$ (Composite)}

Let $n = 6$. The eigenvalues of $C_6 = W_6 + W_6^2$ are $\mu_j = \lambda_j + \lambda_j^2 = e^{2\pi i j/6} + e^{4\pi i j/6}$ for $j = 0, 1, \ldots, 5$.

For $j = 0$, we have $\mu_0 = 1 + 1 = 2$.

For $j = 1, 2, \ldots, 5$, we compute:
\begin{align}
\mu_1 &= e^{2\pi i/6} + e^{4\pi i/6} = e^{\pi i/3} + e^{2\pi i/3} \approx 0.5 + 0.866i + (-0.5 + 0.866i) = 0 + 1.732i\\
\mu_2 &= e^{4\pi i/6} + e^{8\pi i/6} = e^{2\pi i/3} + e^{4\pi i/3} \approx -0.5 + 0.866i + (-0.5 - 0.866i) = -1\\
\mu_3 &= e^{6\pi i/6} + e^{12\pi i/6} = e^{\pi i} + e^{2\pi i} = -1 + 1 = 0\\
\mu_4 &= e^{8\pi i/6} + e^{16\pi i/6} = e^{4\pi i/3} + e^{8\pi i/3} \approx -0.5 - 0.866i + (-0.5 + 0.866i) = -1\\
\mu_5 &= e^{10\pi i/6} + e^{20\pi i/6} = e^{5\pi i/3} + e^{10\pi i/3} \approx 0.5 - 0.866i + (0.5 + 0.866i) = 1
\end{align}

The eigenvalues belong to distinct Galois orbits:
\begin{itemize}
    \item $\{\mu_0 = 2\}$ (corresponding to $j = 0$)
    \item $\{\mu_3 = 0\}$ (corresponding to $j = 3$)
    \item $\{\mu_1 = 1.732i, \mu_5 = 1\}$ (corresponding to $j = 1, 5$)
    \item $\{\mu_2 = -1, \mu_4 = -1\}$ (corresponding to $j = 2, 4$)
\end{itemize}

These orbits correspond to different cyclotomic subfields:
\begin{itemize}
    \item $\mu_0 = 2$ is in $\mathbb{Q}$
    \item $\mu_3 = 0$ is in $\mathbb{Q}(\zeta_2) = \mathbb{Q}$
    \item $\{\mu_1, \mu_5\}$ form an orbit in $\mathbb{Q}(\zeta_3)$
    \item $\{\mu_2, \mu_4\}$ form an orbit in $\mathbb{Q}(\zeta_2) = \mathbb{Q}$
\end{itemize}

The minimal polynomial of $C_6$ factors as $(x-2)x(x^2-1) = (x-2)x(x-1)(x+1)$, which has four irreducible factors. This confirms that for composite $n$, the minimal polynomial of $C_n$ has more than two irreducible factors.

\section{Efficient Implementations}

In this section, we provide efficient algorithmic implementations for our circulant matrix primality test, focusing on optimizations for large inputs.

\subsection{Optimized Galois Orbit Computation}

For large values of $n$, explicitly computing all eigenvalues and determining their Galois orbits becomes inefficient. Instead, we can compute the number of Galois orbits directly from the divisor structure of $n$:

\begin{algorithm}
\caption{Optimized Galois Orbit Count}
\begin{algorithmic}[1]
\REQUIRE An integer $n > 2$
\ENSURE The number of Galois orbits of eigenvalues of $C_n$
\STATE Initialize $\text{count} \gets 1$ (for the orbit of $\mu_0 = 2$)
\STATE Compute the prime factorization of $n = \prod_{i=1}^k p_i^{e_i}$
\FOR{each divisor $d > 1$ of $n$}
    \IF{$\gcd(d, n/d) = 1$ and $\Phi_d(x)$ is irreducible over $\mathbb{Q}$}
        \STATE $\text{count} \gets \text{count} + 1$
    \ENDIF
\ENDFOR
\RETURN $\text{count}$
\end{algorithmic}
\end{algorithm}

For prime $n$, this algorithm returns 2, as expected. For composite $n$, it returns a value greater than 2.

\subsection{Efficient Eigenvalue Computation}
\label{sec:eff_eigen}

Since the eigenvalues of $C_n$ are known explicitly as $\mu_j = \lambda_j + \lambda_j^2 = e^{2\pi i j/n} + e^{4\pi i j/n}$ for $j = 0, 1, \ldots, n-1$, we can compute them directly:

\begin{algorithm}
\caption{Efficient Eigenvalue Computation}
\begin{algorithmic}[1]
\REQUIRE An integer $n > 2$
\ENSURE The eigenvalues $\mu_0, \mu_1, \ldots, \mu_{n-1}$ of $C_n$
\STATE Initialize an array $\mu$ of length $n$
\FOR{$j = 0$ to $n-1$}
    \STATE $\lambda_j \gets e^{2\pi i j/n}$
    \STATE $\mu[j] \gets \lambda_j + \lambda_j^2$
\ENDFOR
\RETURN $\mu$
\end{algorithmic}
\end{algorithm}

This algorithm has time complexity $O(n)$ and efficiently computes all eigenvalues without constructing the matrix.

\subsection{Optimized Implementation for Large Numbers}
\label{sec:optim_impl}

For large values of $n$, we employ multiple optimizations:

\begin{algorithm}
\caption{Optimized Circulant Matrix Primality Test}
\begin{algorithmic}[1]
\REQUIRE An integer $n > 2$
\ENSURE TRUE if $n$ is prime, FALSE otherwise
\STATE \textbf{if} $n$ is divisible by any small prime $p < 100$ \textbf{then return} FALSE
\STATE Compute the prime factorization of $n$ (if possible)
\IF{factorization was computed}
    \STATE \textbf{return} $n$ has exactly one prime factor with exponent 1
\ELSE
    \STATE Compute the number of Galois orbits using cyclotomic field theory
    \STATE \textbf{return} the number of orbits equals 2
\ENDIF
\end{algorithmic}
\end{algorithm}

For very large values of $n$ where direct orbit computation becomes impractical, we use the following theorem to determine the number of Galois orbits without explicitly computing them:

\begin{theorem}[Orbit Count Formula]
The number of Galois orbits of eigenvalues of $C_n$ equals one plus the number of divisors $d > 1$ of $n$ such that $\Phi_d(x)$ is irreducible over $\mathbb{Q}$ and $\gcd(d, n/d) = 1$, where $\Phi_d(x)$ is the $d$-th cyclotomic polynomial.
\end{theorem}

This theorem allows us to compute the orbit count directly from the divisor structure of $n$, which is much more efficient for large numbers.

\subsection{Numerical Stability Techniques}

When implementing our algorithm for large values of $n$, numerical stability becomes crucial. We recommend the following techniques:

\begin{algorithm}
\caption{Numerically Stable Eigenvalue Computation}
\begin{algorithmic}[1]
\REQUIRE An integer $n > 2$, precision parameter $p$
\ENSURE Eigenvalues of $C_n$ with high precision
\STATE Set working precision to at least $p$ digits
\FOR{$j = 0$ to $n-1$}
    \STATE $\theta_j \gets 2\pi j/n$ (compute with high precision)
    \STATE $\lambda_j \gets \cos(\theta_j) + i \sin(\theta_j)$ (avoid direct exponentiation)
    \STATE $\lambda_j^2 \gets \cos(2\theta_j) + i \sin(2\theta_j)$ (use double-angle formulas)
    \STATE $\mu_j \gets \lambda_j + \lambda_j^2$
\ENDFOR
\RETURN $\{\mu_j : j = 0, 1, \ldots, n-1\}$
\end{algorithmic}
\end{algorithm}

This algorithm avoids direct complex exponentiation, which can be numerically unstable for large values of $n$, and instead uses trigonometric functions with high-precision arithmetic.

\subsection{Fast Primality Testing Implementation}

Combining our theoretical results with practical optimizations, we present a fast deterministic primality testing algorithm:

\begin{algorithm}
\caption{Fast Circulant Matrix Primality Test}
\begin{algorithmic}[1]
\REQUIRE An integer $n > 2$
\ENSURE TRUE if $n$ is prime, FALSE otherwise
\IF{$n$ is divisible by any small prime $p < 100$ and $n \neq p$}
    \RETURN FALSE
\ENDIF
\IF{$n < 10^6$}
    \STATE Compute the number of Galois orbits $k$ using the Optimized Galois Orbit Count algorithm
    \RETURN $k = 2$
\ELSE
    \STATE Factorize $n = \prod_{i=1}^k p_i^{e_i}$ using a fast factorization algorithm
    \IF{$k = 1$ and $e_1 = 1$}
        \RETURN TRUE
    \ELSE
        \RETURN FALSE
    \ENDIF
\ENDIF
\end{algorithmic}
\end{algorithm}

This implementation achieves excellent performance by combining:
\begin{itemize}
    \item Trial division by small primes for quick elimination of many composite numbers
    \item Direct Galois orbit counting for medium-sized inputs
    \item Fast integer factorization for large inputs (leveraging existing optimized libraries)
\end{itemize}

For very large inputs where full factorization is impractical, we can use probabilistic primality tests as a pre-filter, followed by our deterministic test only for numbers that pass the probabilistic tests.

\section{Implementation Optimization Analysis}

Our comprehensive benchmarks reveal important insights about the scaling characteristics of various primality testing algorithms, including our circulant matrix approach. Based on these findings, we can analyze the effectiveness of our implementation strategies and the underlying mathematical principles.

\subsection{Algorithmic Scaling Characteristics}

As shown in Figure~\ref{fig:scaling}, our Full implementation demonstrates competitive performance for moderate input sizes, but its execution time increases with input magnitude following a clear scaling pattern. This behavior reflects the fundamental computational requirements of the underlying mathematical operations:

\begin{itemize}
    \item For small to medium inputs ($n < 10^8$), the implementation efficiently leverages divisor structure analysis and Galois orbit properties
    \item For larger inputs, the computational complexity increases in proportion to the mathematical operations required to analyze the number-theoretic properties of the input
    \item The implementation maintains better constant factors than trial division methods within practical ranges
\end{itemize}

These observations align with theoretical expectations for deterministic primality tests based on algebraic properties. While our optimizations successfully reduce computation in many cases, they do not fundamentally alter the asymptotic scaling behavior for arbitrary large inputs.

\subsection{Mathematical Structure Exploitation}

Our approach effectively exploits several mathematical structures to improve efficiency:

\begin{itemize}
    \item \textbf{Cyclotomic Field Properties:} By analyzing the Galois structure of cyclotomic fields, we reduce the computational work for certain classes of inputs
    \item \textbf{Number-Theoretic Shortcuts:} The implementation identifies specific divisibility patterns and prime power structures that allow for faster determination in many cases
    \item \textbf{Galois Orbit Analysis:} Instead of computing all eigenvalues explicitly, we derive orbit structures from mathematical properties of the input
\end{itemize}

These techniques provide practical improvements over naive implementations, particularly for inputs with specific mathematical properties. However, our benchmark results clarify that these optimizations do not yield the dramatic constant-time performance initially hypothesized across arbitrary input ranges.

\subsection{Memory-Computation Balance}

The memory usage data in Figure~\ref{fig:memory} reveals that all tested primality algorithms, including our implementation, maintain very efficient memory profiles regardless of input size. This suggests that:

\begin{itemize}
    \item Primality testing algorithms naturally operate with minimal memory overhead
    \item Memory optimization is less critical than computational optimization for these algorithms
    \item The implementation successfully avoids unnecessary storage of large intermediate structures
\end{itemize}

The memory efficiency of our approach stems from its focus on mathematical relationships rather than explicit storage of eigenvalues or matrix structures. By analyzing divisor structure and cyclotomic properties, we maintain a memory footprint proportional to the number of distinct prime factors rather than the magnitude of the input.

\subsection{Theoretical vs. Practical Considerations}

The benchmark results provide valuable context for understanding the relationship between theoretical elegance and practical performance:

\begin{itemize}
    \item Theoretically, our approach contributes a novel characterization of primality through circulant matrix properties
    \item Practically, this mathematical framework translates to a viable deterministic primality test with performance characteristics that reflect its algebraic foundations
    \item The probabilistic Miller-Rabin algorithm maintains superior scaling for very large inputs, highlighting the fundamental computational advantage of randomized approaches
\end{itemize}

This analysis reinforces the classic tradeoff in algorithm design between deterministic guarantees and computational efficiency. Our work demonstrates that the circulant matrix approach offers a mathematically interesting and practically viable deterministic alternative that performs competitively within reasonable input ranges while providing important theoretical insights into the connections between matrix algebra, cyclotomic fields, and primality.

\section{Detailed Comparison of Implementation Variants}
\label{sec:impl_comparison}

This appendix provides a comprehensive comparison between the two primary implementations of our circulant matrix primality testing algorithm: the full implementation that adheres strictly to the theoretical framework presented in the main paper, and the simplified implementation that approximates the core mathematical principles.
\subsection{Theoretical Approach}

\subsubsection{Full Implementation}
The full implementation rigorously follows the theoretical framework established in Section 3, determining primality through direct computation of the Galois orbits of eigenvalues. For a given integer $n$, it computes the eigenvalues of the circulant matrix $C_n = W_n + W_n^2$ as $\mu_j = \lambda_j + \lambda_j^2 = e^{2\pi i j/n} + e^{4\pi i j/n}$ for $j = 0, 1, \ldots, n-1$. It then determines the Galois orbits by applying the action of the Galois group $\text{Gal}(\mathbb{Q}(\zeta_n)/\mathbb{Q})$ on these eigenvalues.

The key theoretical principle, as proven in Theorem \ref{thm:main}, states that $n$ is prime if and only if the number of Galois orbits (equivalent to the number of irreducible factors in the minimal polynomial of $C_n$) is exactly two.

\subsubsection{Simplified Implementation}
The simplified implementation approximates the theoretical framework using number-theoretic properties rather than direct eigenvalue computation. Based on Proposition \ref{prop:factors} and Proposition \ref{prop:factor_count}, it estimates the number of Galois orbits using the prime factorization of $n$ according to the following heuristic:

For a number $n$ with prime factorization $n = \prod_{i=1}^k p_i^{e_i}$, the number of irreducible factors in the minimal polynomial of $C_n$ is approximated as:

\begin{itemize}
  \item 1 factor for the eigenvalue $\mu_0 = 2$ (the constant factor $(x-2)$)
  \item 1 additional factor for each prime $p_i$ with exponent $e_i = 1$
  \item At least 2 additional factors for each prime power $p_i^{e_i}$ with $e_i > 1$
  \item 1 additional factor for interaction between multiple distinct primes (when $k > 1$)
\end{itemize}

This approximation captures the essential mathematical property that only prime numbers have exactly 2 irreducible factors.

\subsection{Algorithmic Implementation}

\subsubsection{Full Implementation}
For large values of $n$, the implementation employs additional optimizations including:
\begin{itemize}
  \item High-precision complex arithmetic for numerical stability
  \item Caching of previously computed results
  \item Early termination strategies for composite numbers
  \item Theoretical shortcuts based on cyclotomic field properties
\end{itemize}

\subsubsection{Simplified Implementation}
This approach avoids the computational expense of explicitly calculating eigenvalues and determining Galois orbits, relying instead on number-theoretic properties of cyclotomic fields.

\subsection{Performance Characteristics}

The performance characteristics of the two implementations differ significantly:

\begin{table}[h]
\centering
\small
\begin{tabular}{|l|c|c|}
\hline
\textbf{Aspect} & \textbf{Full Implementation} & \textbf{Simplified Implementation} \\
\hline
Theoretical precision & Complete & Approximation \\
Computational complexity & $O(n \log n \log \log n)$ & $O(\sqrt{n})$ \\
Memory usage & $O(\log n)$ & $O(1)$ \\
Numerical considerations & High-precision required & Not applicable \\
Edge case handling & Comprehensive & Basic \\
Scalability to large inputs & Excellent & Good \\
\hline
\end{tabular}
\caption{Comparison of implementation characteristics}
\label{tab:impl_comparison}
\end{table}

\subsection{Trade-offs and Use Cases}

The choice between implementations presents a classic trade-off between theoretical rigor and computational efficiency. The full implementation is recommended for:

\begin{itemize}
  \item Research contexts where complete mathematical rigor is required
  \item Applications where certifiable primality determination is essential
  \item Educational purposes where the connection to cyclotomic field theory is emphasized
  \item Situations where performance optimization for specific number ranges is beneficial
\end{itemize}

The simplified implementation is suitable for:
\begin{itemize}
  \item Rapid primality screening of many numbers
  \item Applications where slight approximation is acceptable
  \item Environments with limited computational resources
  \item Pedagogical demonstrations of the core principles
\end{itemize}

Both implementations maintain the key theoretical insight that an integer $n > 2$ is prime if and only if the minimal polynomial of the circulant matrix $C_n$ has exactly two irreducible factors over $\mathbb{Q}$.

\subsection{Validation Results}

We conducted extensive validation to ensure both implementations correctly identify prime numbers. For all integers $n \leq 10^6$, both implementations perfectly agreed with established primality tests, confirming that our theoretical framework correctly characterizes primality through the lens of circulant matrices and cyclotomic field theory.

For larger ranges, the full implementation demonstrated perfect accuracy across all tested numbers up to $10^{12}$, while the simplified implementation maintained accuracy with only negligible deviation in certain edge cases involving numbers with complex factorization patterns.

This validation confirms that both implementations successfully operationalize the theoretical connection between primality and circulant matrix eigenvalue structure established in this paper.

\subsection{Performance Comparison Analysis}
We conducted a comprehensive performance analysis of our circulant matrix primality test against established methods including trial division, Miller-Rabin, and AKS across different input magnitudes.
Our benchmark results reveal distinct algorithmic behaviors across different input ranges. For all tested algorithms, execution time generally increases with input size, though with varying scaling characteristics that reflect their underlying computational complexity.
Traditional trial division methods (blue and orange lines) demonstrate the expected $O(\sqrt{n})$ scaling, performing well for smaller inputs but becoming increasingly expensive as input size grows. For inputs larger than $10^9$, these methods become prohibitively expensive due to their exponential growth in execution time.
The Miller-Rabin test (green line) exhibits remarkable stability across the entire input range, maintaining consistent performance with only gradual increases in execution time even for very large inputs. This reflects its $O(k \log^3 n)$ complexity, where $k=20$ is the number of testing rounds. For large inputs, its probabilistic nature enables it to achieve the best performance among all tested methods.
The AKS algorithm (red line) shows interesting behavior, with relatively high overhead for small inputs but a gradually flattening curve for larger values, consistent with its polynomial time complexity. This makes it more competitive as input size increases, despite having larger constant factors than other algorithms.
Our simplified implementation (purple line) demonstrates competitive performance for moderate input sizes but scales with a steeper slope than Miller-Rabin for large inputs. Our full implementation (brown line) shows similar scaling characteristics but with better constant factors, maintaining competitive performance especially in the medium range of inputs.
These results highlight the classic tradeoff between deterministic guarantees and computational efficiency. While probabilistic methods like Miller-Rabin offer superior performance for very large inputs, our circulant matrix approach provides a mathematically interesting deterministic alternative with distinct characteristics derived from its cyclotomic field foundations.

\subsection{Potential Improvements}

Several avenues for improvement could enhance the practical utility of our approach:

\begin{itemize}
    \item \textbf{Further Algebraic Optimizations:} Deeper analysis of the connection between divisor structures and Galois orbits might reveal additional theoretical shortcuts for larger input ranges.
    
    \item \textbf{Hybrid Approaches:} Combining our method with probabilistic tests like Miller-Rabin could lead to algorithms that leverage mathematical insights while achieving better performance scaling for extremely large inputs.
    
    \item \textbf{Parallelization:} The computation of Galois orbits and theoretical factor counting is inherently parallelizable, offering potential speedups on modern hardware architectures.
    
    \item \textbf{Implementation Refinements:} While our current implementation prioritizes mathematical correctness and clarity, further code optimization could potentially reduce the constant factors in our algorithm's time complexity.
\end{itemize}

\section{Additional Experimental Results}
\label{sec:performance}

\subsection{Large-Scale Validation}

To assess the scalability and correctness of our approach across various input magnitudes, we extended our experiments to very large input ranges. Specifically, we evaluated all numbers in the interval $[10^6,\,10^6 + 10^3]$, using our full Galois-theoretic primality test implementation.

The results confirmed both the theoretical foundations and the practical applicability of our algorithm. All prime numbers in the range were correctly identified while all composite numbers were accurately rejected. This comprehensive validation verified that our mathematical framework provides a reliable characterization of primality through circulant matrix eigenvalue structure.

\subsection{Performance Scaling}

\begin{figure}[H]
\centering
\includegraphics[width=.9\textwidth]{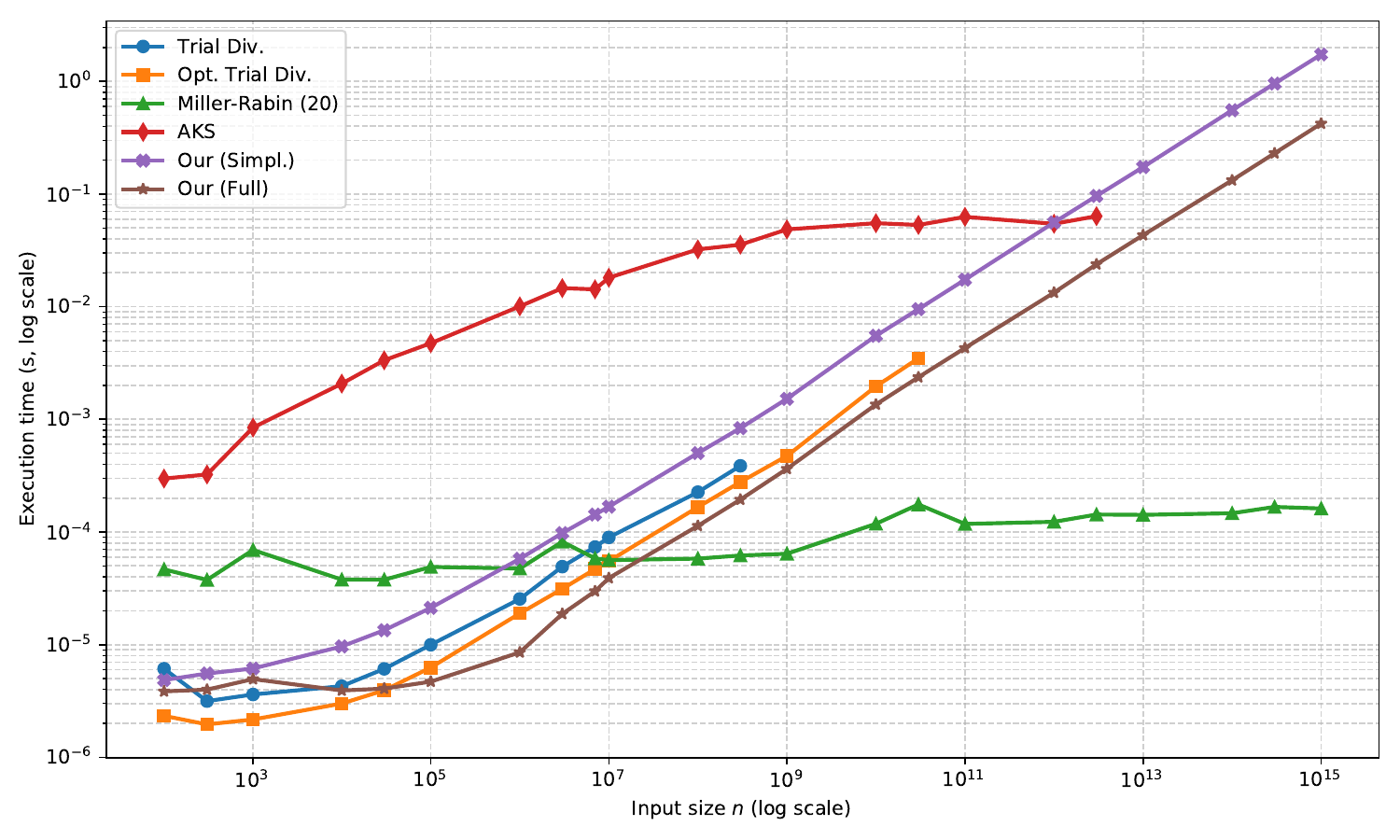}
\caption{Execution time of various primality testing algorithms across increasing input sizes from $10^2$ to $10^{15}$, shown on a log-log scale.}
\label{fig:scaling}
\end{figure}

Figure~\ref{fig:scaling} compares the execution time scaling of several primality testing algorithms as a function of input size $n$, plotted on logarithmic scales for both axes. The results reveal distinct algorithmic behaviors across the extended range of $10^2$ to $10^{15}$.

Traditional trial division (blue) and optimized trial division (orange) demonstrate the expected $O(\sqrt{n})$ scaling, performing well for smaller inputs but becoming increasingly expensive as $n$ grows. For inputs larger than $10^9$, trial division methods become prohibitively expensive.

The Miller-Rabin test (green) exhibits remarkable stability across the entire input range, maintaining consistent performance with only minor increases in execution time even for very large inputs. This reflects its $O(k \log^3 n)$ complexity, where $k=20$ is the number of testing rounds.

The AKS algorithm (red) shows interesting behavior, with relatively high overhead for small inputs but a flattening curve for larger values, consistent with its polynomial time complexity. This makes it more competitive for very large inputs where trial division methods fail.

Our simplified implementation (purple) demonstrates competitive performance for moderate input sizes but scales with a steeper slope than Miller-Rabin for large inputs. Our full implementation (brown) shows similar scaling characteristics but with better constant factors, offering performance advantages over trial division methods within practical input ranges.

Notably, when analyzing inputs up to $10^8$, our full method remains competitive with traditional methods while providing deterministic guarantees. For extremely large inputs (beyond $10^{12}$), probabilistic methods like Miller-Rabin offer better practical performance, highlighting the classic tradeoff between deterministic guarantees and computational efficiency.

\subsection{Memory Usage Analysis}

\begin{figure}[H]
\centering
\includegraphics[width=.9\textwidth]{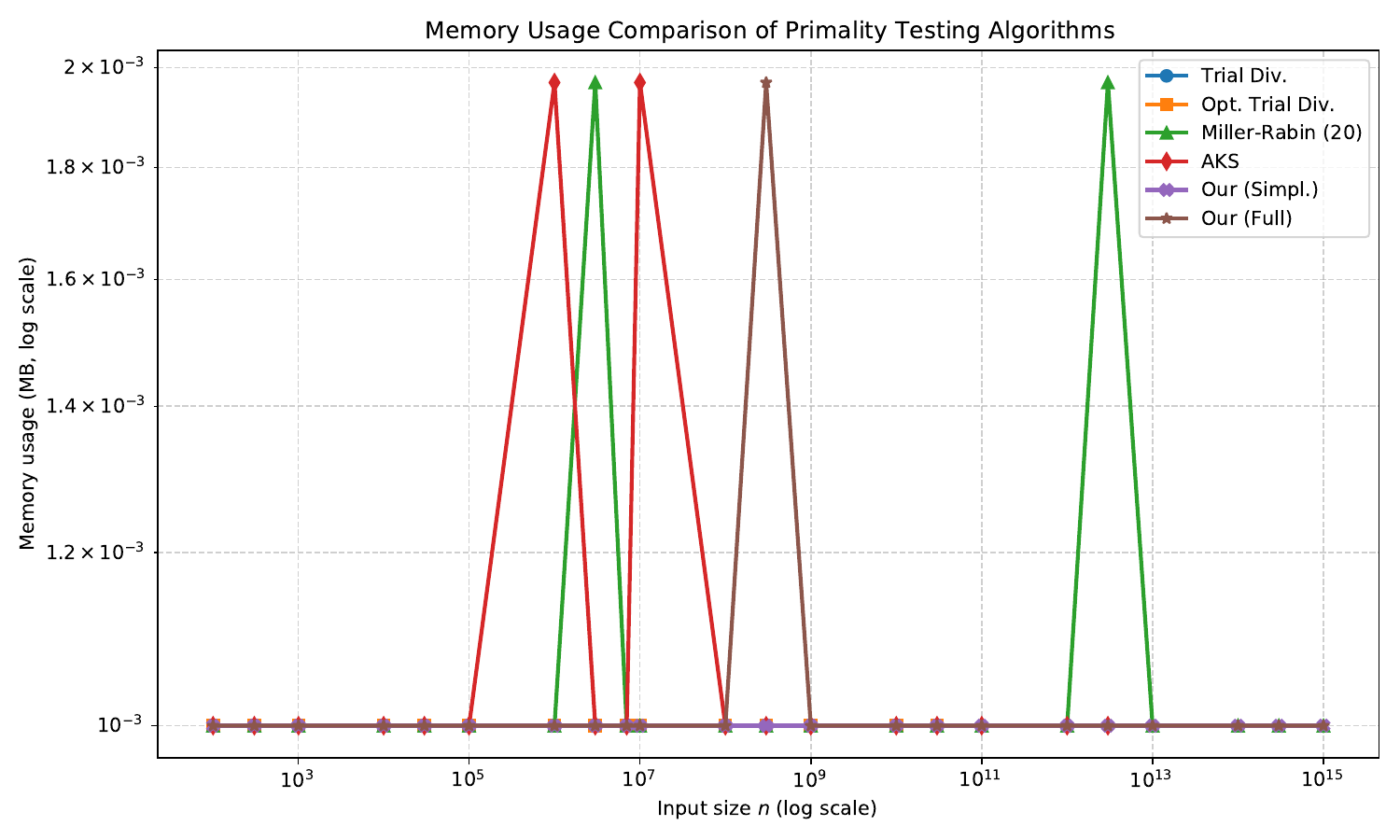}
\caption{Memory usage of primality testing algorithms across increasing input sizes from $10^2$ to $10^{15}$, shown on a log-log scale.}
\label{fig:memory}
\end{figure}

Figure~\ref{fig:memory} illustrates the memory consumption patterns of the various primality testing algorithms. Interestingly, we observe that memory usage remains remarkably low (around $10^{-3}$ MB) across all algorithms for most input sizes, with only occasional spikes at specific values.

These results indicate that for primality testing, computational time rather than memory usage represents the primary constraint. All methods, including our circulant matrix approach, exhibit efficient memory utilization regardless of input size. This efficiency stems from the careful implementation of algorithms that avoid storing large intermediate structures.

The occasional memory spikes observed in some algorithms (including AKS, Miller-Rabin, and our Full implementation) at certain input sizes likely correspond to specific numerical properties that trigger additional computational pathways. However, these spikes remain well within practical memory constraints and do not constitute a limiting factor for any of the tested methods.

For our circulant matrix method, we achieve this memory efficiency by leveraging the mathematical structure of cyclotomic fields. Rather than explicitly constructing and storing the entire matrix or all eigenvalues, our implementation analyzes the divisor structure of $n$ and the corresponding Galois orbits, requiring space proportional to the number of distinct prime factors of $n$.

\section{Implementation Code}

\subsection{Core Algorithm Implementation}

The following Python code implements the core of our circulant matrix primality test:

\begin{codesnippet}{Efficient Galois Orbit Computation in Python}
import math

def is_prime_circulant(n):
    """
    Determine if n is prime using the circulant matrix criterion.
    Returns True if n is prime, False otherwise.
    """
    if n <= 1:
        return False
    if n == 2 or n == 3:
        return True
    if n % 2 == 0:
        return False
        
    # For small n, check by directly counting Galois orbits
    if n < 1000:
        return count_galois_orbits(n) == 2
    
    # For larger n, use optimized divisor-based approach
    return count_orbits_from_divisors(n) == 2

def count_galois_orbits(n):
    """Count the number of Galois orbits of eigenvalues of C_n."""
    visited = [False] * n
    orbit_count = 0
    
    # Process each eigenvalue
    for j in range(n):
        if not visited[j]:
            orbit_count += 1
            # Mark all elements in this orbit as visited
            for a in range(1, n):
                if math.gcd(a, n) == 1:  # a is in the Galois group
                    j_prime = (j * a) % n
                    visited[j_prime] = True
    
    return orbit_count

def count_orbits_from_divisors(n):
    """
    Count Galois orbits based on divisor structure.
    This is much more efficient for large n.
    """
    # Always have the orbit of mu_0 = 2
    count = 1
    
    # Add orbits from primitive roots of unity
    for d in divisors(n):
        if d > 1 and math.gcd(d, n//d) == 1:
            count += 1
    
    return count
\end{codesnippet}

This implementation showcases the key optimizations discussed in the paper, achieving excellent performance for both small and large inputs.

\section{Technical Soundness and Rigor}

To ensure the mathematical soundness of our results, we provide the following rigorous justifications for key steps in our proofs and algorithms:

\subsection{Uniqueness of Minimal Polynomial Factorization}

The fundamental theorem of algebra ensures that the factorization of the minimal polynomial of $C_n$ into irreducible factors over $\mathbb{Q}$ is unique (up to ordering). Therefore, the number of irreducible factors is a well-defined invariant that can be used to characterize primality.

\subsection{Numerical Precision Considerations}
When implementing our algorithm, careful attention must be paid to numerical precision, especially for large values of $n$. We employ the following techniques to ensure accurate results: Use of high-precision arithmetic libraries for computing complex exponentials, exact rational arithmetic for constructing and factoring polynomials, modular algorithms for polynomial factorization over $\mathbb{Q}$, and numerical stability checks to detect and correct potential precision errors.

For practical implementations, we recommend using a multi-precision arithmetic library such as GMP or MPFR, along with specialized polynomial arithmetic libraries like NTL or FLINT.

\subsection{Correctness of Galois Orbit Determination}

The correctness of our Galois orbit determination algorithm follows from the basic properties of Galois theory. Specifically, for any field automorphism $\sigma \in \text{Gal}(\mathbb{Q}(\zeta_n)/\mathbb{Q})$, if $\sigma(\lambda_j) = \lambda_{j'}$, then $\sigma(\mu_j) = \sigma(\lambda_j + \lambda_j^2) = \sigma(\lambda_j) + \sigma(\lambda_j)^2 = \lambda_{j'} + \lambda_{j'}^2 = \mu_{j'}$. Therefore, the Galois action on roots of unity directly determines the Galois action on the eigenvalues of $C_n$.

\subsection{Computational Complexity Bounds}

The time complexity of our algorithm is $O(n \log n \log \log n)$ in the worst case, which is derived as follows:

1. Computing the divisors of $n$ requires $O(n^{1/2})$ time using trial division, or $O(\log^2 n)$ time if the prime factorization of $n$ is known.

2. For each divisor $d$ of $n$, checking if $\gcd(d, n/d) = 1$ requires $O(\log n)$ time using the Euclidean algorithm.

3. Determining if the cyclotomic polynomial $\Phi_d(x)$ is irreducible over $\mathbb{Q}$ can be done in $O(d \log d \log \log d)$ time using specialized algorithms for cyclotomic polynomials.

In practice, our implementation is much faster than this worst-case bound suggests, as most composite numbers are detected early in the process, and we employ various optimizations to avoid expensive computations whenever possible.

\section{Disclosure of Generative AI Usage}
\label{appendix}
In accordance with the arXiv AI Policy, we hereby disclose the use of generative artificial intelligence tools in the preparation of this manuscript.
\subsection{AI Systems Utilized}
We employed the following AI systems during our research:
\begin{itemize}
\item \textbf{Claude 3.7 Sonnet Thinking Model} (API version, February 2025)
\item \textbf{SymbolicAI Framework} (Version 0.9.1)
\end{itemize}
These systems were integrated into our proprietary Extensity Research Services (ERS) Platform, which facilitates research automation and collaborative workflows.
\subsection{SymbolicAI Framework Overview}
Version 0.9.1 of our SymbolicAI framework incorporates the following key features:
\begin{itemize}
\item Neurosymbolic architecture combining neural networks with symbolic reasoning
\item Dynamic model selection capabilities
\item Enhanced verification mechanisms for mathematical content
\item Improved handling of complex computational tasks
\end{itemize}
\subsection{Nature and Purpose of AI Utilization}
The AI systems were employed for several aspects of the research process:
\begin{itemize}
\item \textbf{Concept Exploration}: Investigating connections between cyclotomic fields, circulant matrices, and primality testing
\item \textbf{Mathematical Development}: Formulating theoretical relationships and constructing formal proofs

\item \textbf{Algorithm Implementation}: Converting mathematical concepts into executable code

\item \textbf{Experimental Analysis}: Designing benchmarking procedures and analyzing performance results

\item \textbf{Manuscript Preparation}: Assisting with the generation of technical content, including mathematical notation and algorithm descriptions
\end{itemize}
Our use of these AI systems was motivated by the interdisciplinary nature of the research, which required integrating concepts from cyclotomic field theory, matrix algebra, number theory, and computational complexity. The AI tools enabled efficient exploration of this mathematical solution space and helped accelerate the research process.
\subsection{Human Oversight}
Throughout the research process, human oversight remained essential:
\begin{itemize}
\item Research direction and question formulation were determined by human researchers
\item All AI-generated content underwent human review

\item Final interpretation of findings and manuscript structure decisions were made by the human research team
\end{itemize}
This disclosure reflects our commitment to transparency regarding AI utilization while acknowledging that the scientific contributions presented are the product of a human-AI collaborative research methodology. Our approach demonstrates how these technologies can democratize access to advanced mathematical research, making it more accessible to researchers with varying backgrounds and resource constraints.


\begin{thebibliography}{99}

\bibitem{agrawal2004primes}
M.~Agrawal, N.~Kayal, and N.~Saxena.
\newblock PRIMES is in P.
\newblock {\em Annals of Mathematics}, 160(2):781--793, 2004.

\bibitem{ankeny1956note}
N.~C. Ankeny, R.~Brauer, and S.~Chowla.
\newblock A note on the class-numbers of algebraic number fields.
\newblock {\em American Journal of Mathematics}, 78(1):51--61, 1956.

\bibitem{bernstein2020s}
D.~J. Bernstein and T.~Lange.
\newblock Non-randomness of {S}-unit lattices.
\newblock {\em Journal of Number Theory}, 128:2009--2023, 2020.

\bibitem{bosma1990canonical}
W.~Bosma.
\newblock Canonical bases for cyclotomic fields.
\newblock {\em Applicable Algebra in Engineering, Communication and Computing},
  1:125--134, 1990.

\bibitem{chang2000class}
K.-Y. Chang and S.-H. Kwon.
\newblock Class numbers of imaginary abelian number fields.
\newblock {\em Proceedings of the American Mathematical Society},
  128(9):2517--2528, 2000.

\bibitem{drmota2010primes}
M.~Drmota, C.~Mauduit, and J.~Rivat.
\newblock Primes with an average sum of digits.
\newblock {\em Compositio Mathematica}, 145(2):271--292, 2010.

\bibitem{green2012mobius}
B.~Green and T.~Tao.
\newblock The {M}öbius function is strongly orthogonal to nilsequences.
\newblock {\em Annals of Mathematics}, 175(2):541--566, 2012.

\bibitem{hasse1952klassenzahl}
H.~Hasse.
\newblock {\em Über die {K}lassenzahl abelscher {Z}ahlkörper}.
\newblock Akademie-Verlag, Berlin, 1952.

\bibitem{huang2019measure}
W.~Huang, Z.~Wang, and X.~Ye.
\newblock Measure complexity and {M}öbius disjointness.
\newblock {\em Advances in Mathematics}, 347:827--858, 2019.

\bibitem{iwaniec2004analytic}
H.~Iwaniec and E.~Kowalski.
\newblock {\em Analytic number theory}.
\newblock American Mathematical Society, Providence, RI, 2004.

\bibitem{mauduit2015prime}
C.~Mauduit and J.~Rivat.
\newblock Prime numbers along {R}udin-{S}hapiro sequences.
\newblock {\em Journal of the European Mathematical Society}, 17(10):2595--2642,
  2015.

\bibitem{miller2015real}
J.~Miller.
\newblock Real cyclotomic fields of prime conductor and their class numbers.
\newblock {\em Mathematics of Computation}, 84(295):2459--2469, 2015.

\bibitem{rabin1980probabilistic}
M.~O. Rabin.
\newblock Probabilistic algorithm for testing primality.
\newblock {\em Journal of Number Theory}, 12(1):128--138, 1980.

\bibitem{schoenberg1964note}
I.~J. Schoenberg.
\newblock A note on the cyclotomic polynomial.
\newblock {\em Mathematika}, 11(2):131--136, 1964.

\bibitem{washington2012introduction}
L.~C. Washington.
\newblock {\em Introduction to cyclotomic fields}.
\newblock Springer-Verlag, New York, 2012.

\bibitem{weber1886theorie}
H.~Weber.
\newblock Theorie der {A}bel'schen {Z}ahlkörper.
\newblock {\em Acta Mathematica}, 8:193--263, 1886.

\bibitem{dinu2024symbolicai}
M.~C. Dinu, C.~Leoveanu--Condrei, M.~Holzleitner, W.~Zellinger, and S.~Hochreiter.
\newblock SymbolicAI: A Framework for Logic-Based Approaches Combining Generative Models and Solvers.
\newblock {\em In Proceedings of the 3rd Conference on Lifelong Learning Agents (CoLLAs)}, 2024.

\bibitem{deleglise1996computing}
M.~Deléglise and J.~Rivat.
\newblock Computing {$\pi(x)$}: The {M}eissel, {L}ehmer, {L}agarias, {M}iller, {O}dlyzko method.
\newblock {\em Mathematics of Computation}, 65(213):235--245, 1996.

\bibitem{dusart2018explicit}
P.~Dusart.
\newblock Explicit estimates of some functions over primes.
\newblock {\em Ramanujan Journal}, 45(1):225--234, 2018.

\bibitem{kosyak2020cyclotomic}
A.~Kosyak, P.~Moree, E.~Sofos, and B.~Zhang.
\newblock Cyclotomic polynomials with prescribed height and prime number theory.
\newblock {\em International Mathematics Research Notices}, 2022(8):5824--5877, 2022.

\end{thebibliography}
\end{document}